\documentclass[aps,amsmath,amssymb, notitlepage, twocolumn,superscriptaddress,longbibliography,nofootinbib,floatfix]{revtex4-2}
\usepackage{xcolor}
\definecolor{darkred}  {rgb}{0.5,0,0}
\definecolor{darkblue} {rgb}{0,0,0.5}
\definecolor{darkgreen}{rgb}{0,0.5,0}
\usepackage[pdftex,colorlinks=true,linkcolor=darkblue,citecolor=blue,urlcolor=darkred]{hyperref}
\usepackage{braket}
\usepackage{amsmath,amsfonts, amssymb, amsthm, dsfont}
\usepackage{yfonts}
\usepackage{bm,breqn}
\usepackage{mathrsfs}
\usepackage{svg}
\usepackage{graphicx}
\usepackage[caption=false]{subfig}
\usepackage{verbatim,soul}
\usepackage{wasysym}
\usepackage{enumerate}

\usepackage{tikz}
\usetikzlibrary{calc}
\usepackage{multirow}
\usepackage[capitalise]{cleveref}

\theoremstyle{definition}
\newtheorem{theorem}{Theorem}
\newtheorem{lemma}{Lemma}
\newtheorem*{lemma*}{Lemma}

\newtheorem{corollary}{Corollary}

\theoremstyle{remark}

\graphicspath{{figs/}{}}

\newlength{\fighskip} \fighskip=2pt
\newlength{\figvskip} \figvskip=3pt
\newcommand*{\figbox}[2]{{
		\def\figscale{#1}
		\def\arraystretch{0.8}
		\arraycolsep=0pt
		\begin{array}{c}
			\vbox{\vskip\figscale\figvskip
				\hbox{\hskip\figscale\fighskip
					\includegraphics[scale=\figscale]{#2}}}
\end{array}}}

\DeclareMathOperator{\Tr}{Tr}
\DeclareMathOperator{\disc}{disc}
\DeclareMathOperator{\supp}{supp}
\newcommand{\norm}[1]{\left\lVert #1 \right\rVert}

\newcommand{\refeq}[1]{Eq.~(\ref{#1})}

\newcommand{\ketbra}[1]{\ket{#1}\bra{#1}}

\theoremstyle{definition}
\newtheorem{result}{Result}

\newenvironment{result_restate}[2]
{\renewcommand{\theresult}{\ref{#1}#2}
	\addtocounter{result}{-1}
	\begin{result}}
	{\end{result}}
\newenvironment{lemma_copy}[1]
{\renewcommand{\thelemma}{\ref{#1}}
	\addtocounter{lemma}{-1}
	\begin{lemma}}
	{\end{lemma}}

\begin{document}
	
	\title{How Much Entanglement Is Needed for Topological Codes and Mixed States with Anomalous Symmetry?}
	\newcommand*{\PI}{Perimeter Institute for Theoretical Physics, Waterloo, Ontario N2L 2Y5, Canada}
	\newcommand*{\UW}{Department of Physics and Astronomy, University of Waterloo, Waterloo, Ontario N2L 3G1, Canada}
	\newcommand*{\NRC}{National Research Council Canada, Waterloo, Ontario N2L 3W8, Canada}
	
	\author{Zhi Li}\affiliation{\PI}       \affiliation{\NRC}   
	\author{Dongjin Lee}\affiliation{\PI}
	\affiliation{\UW}
	\author{Beni Yoshida}\affiliation{\PI}

	\begin{abstract}
		It is known that particles with exotic properties can emerge in systems made of simple constituents such as qubits, due to long-range quantum entanglement. In this paper, we provide quantitative characterizations of entanglement necessary for emergent anyons and fermions by using the geometric entanglement measure (GEM), which quantifies the maximal overlap between a given state and any short-range-entangled states. For systems with emergent anyons, based on the braiding statistics, we show that the GEM scales linearly in the system size regardless of microscopic details. The phenomenon of emergent anyons can also be understood within the framework of quantum error correction (QEC). Specifically, we show that the GEM of any 2D stabilizer codes must be at least quadratic in the code distance. Our proof is based on a generic prescription for constructing string operators, establishing a rigorous and direct connection between emergent anyons and QEC. 
		For systems with emergent fermions, despite that the ground state subspaces could be exponentially huge and their coding properties could be rather poor, we show that the GEM also scales linearly in the system size. 
		Our analysis establishes an intriguing link between quantum anomaly and entanglement: A quantum state respecting anomalous $1$-form symmetries, must be long-range-entangled and have large GEM. Our results also extend to mixed states, such as the $ZX$-dephased toric code, providing a provably nontrivial class of intrinsically mixed-state phases.
		
	\end{abstract}
	\maketitle

	\section{Introduction}
	
	Quantum information processing is intrinsically vulnerable to noises and decoherence that may arise from physical interactions with the environment as well as engineering imperfections in handling qubits. 
	A useful way to fault-tolerantly store and process quantum information is to utilize particles with exotic statistical properties known as anyons \cite{leinaas1977theory,wilczek1982magnetic}, which provide ideal platforms with naturally arising topological protection from otherwise detrimental errors~\cite{Kitaev:1997wr,dennis2002topological}.
	While anyons may not naturally exist as elementary particles in nature, they do \emph{emerge} as low-energy excitations in strongly correlated many-body quantum systems---most notably in fractional quantum Hall states \cite{laughlin1983anomalous}---demonstrating their physical reality.
	Theoretical and experimental investigations of anyons remain one of the most active areas of research in modern physics.
	
	The emergent phenomenon of anyons is closely related to long-range entanglement in many-body systems~\cite{bravyi2006lieb,kitaev2006topological,levin2006detecting,zhang2012quasiparticle,kim2015ground,shi2020verlinde,shi2020fusion}.
	A vivid feature of quantum entanglement is the superposition: Construction of wave functions requires a large number of product basis states.
	For instance, consider a ground state of the 2D toric code. It can be expressed as a sum of exponentially many computational ($Z$-basis) states 
	\begin{align}
		|\Psi_{\text{toric}}\rangle = \frac{1}{\sqrt{2^{n/2}}} \prod_{\text{vertex}}
		(I + \figbox{0.22}{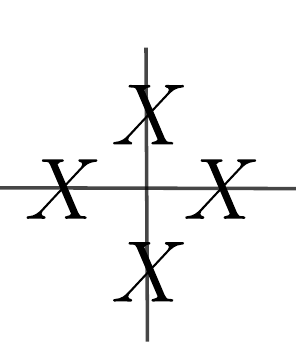} ) |0\rangle^{\otimes n} 
		\propto \sum_{\alpha \in \text{closed-loop}} | \alpha \rangle
	\end{align}
	where each term can be viewed as a closed-loop configuration by drawing a line connecting qubits in $|1\rangle$. 
	This characterization of the toric code is an example of the string-net condensation picture which hosts a wide generalization for nonchiral topological order~\cite{levin2005string}. 
	Anyonic excitations can be generally created at end points of an open string operator $\gamma$ applied to the entangled ground state.
	Nontrivial braiding statistics of anyons implies that the ground state cannot be prepared by a short-depth quantum circuit, and, thus, must be long-range-entangled~\cite{bravyi2006lieb,haah2016invariant,aharonov2018quantum}.
	
	The degree of superpositions in a given wave function $|\Psi\rangle$ can be quantitatively measured by the geometric entanglement measure (GEM) of the form~\cite{wei2003geometric, botero2007scaling, orus2008universal,orus2008geometric,wei2010entanglement, orus2014geometric}
	\begin{align}
		E(\Psi) = - \max_{ |P\rangle \, :\,  \text{product}} \log_2 |\langle \Psi | P \rangle|^2
	\end{align}
	which computes the maximum overlap between $|\Psi\rangle$ and product states. 
	Unfortunately, this measure fails to distinguish short-range and long-range entanglement, since short-range-entangled states, such as the cluster state, can score high on the GEM. 
	To circumvent this problem, the depth-$t$ geometric entanglement measure (depth-$t$ GEM) was proposed~\cite{bravyi2024much}
	\begin{align}
		E_{t}(\Psi) =  -\max_{U\, : \, \mathrm{depth}(U)= t}\;   \log_2{|\langle \Psi| U\ket{0}^{\otimes n}|^2}
	\end{align}
	which considers the maximum overlap between $|\Psi\rangle$ and all short-range-entangled states $U|0\rangle^{\otimes n}$.
	Here, the maximum is taken over all $n$-qubit geometrically local (e.g., brick wall) depth-$t$ quantum circuits; a depth-$0$ circuit is defined as a product of single-qubit gates, and, hence, $E(\Psi)=E_{0}(\Psi)$.
	
	This generalization of GEM provides a refined characterization of long-range entanglement. 
	According to the conventional definition, a quantum state $|\Psi\rangle$ is said to be long-range-entangled if it cannot be (approximately) prepared by a short-depth circuit. 
	In terms of the GEM, this simply amounts to $E_{t}(\Psi)\not=0$ [or $E_{t}(\Psi)>\epsilon$ if approximations are accounted] for small $t$.
	However, depending on its underlying entanglement structure, each long-range-entangled state $|\Psi\rangle$ can have strikingly different values of $E_{t}(\Psi)$.
	Some long-range-entangled states score very low on the GEM, suggesting that they have a significant overlap with short-range-entangled states despite their high quantum circuit complexity. 
	On the contrary, genuinely long-range-entangled states score very high on the GEM, 
	suggesting that they not only cannot be approximated by short-range-entangled states, but also must be almost orthogonal to any short-range-entangled states. 
	For instance, the Greenberger-Horne-Zeilinger (GHZ) state $\frac{1}{\sqrt{2}}(|0\rangle^{\otimes n} + |1\rangle^{\otimes n})$ is long-range-entangled, since it cannot be prepared from a product state by a constant depth circuit. 
	However, the GHZ state scores small for depth-$t$ GEM, namely, $E_{t}(\text{GHZ}) \sim O(1)$ for any $t \sim O(1)$. 
	This aligns with our intuition concerning the intrinsic instability of entanglement in the GHZ state. 
	On the other hand, as we prove later in this paper, the toric code has $E_{t}(\text{toric})\sim O(n)$ for $t \sim O(1)$, reflecting its complex entanglement structure. 
	
	In this work, we use the depth-$t$ GEM as a more direct, quantitative characterization of quantum entanglement that is needed for emergent anyons and other exotic physical objects. 
	
	\subsection{Emergent anyons}
	
	In this paper, we begin by studying two-dimensional spin systems whose ground states support anyonic excitations with nontrivial braiding statistics. 
	While we focus on the toric code for simplicity of presentation, similar results can be generically derived. We use $n$, $k$, and $d$ to refer the number of spins, the logarithmic ground state degeneracy, and the code distance, respectively. 
	
	\begin{result}\label{result:toric}
		A ground state $|\Psi\rangle$ of the toric code satisfies
		\begin{align}
			E_{t}(\Psi) =\Omega(n) ,
			\label{eq:result}
		\end{align}
		where\footnote{In this paper, we use standard Bachmann-Landau notation~\cite{knuth1976big} for $\Omega$, $\Theta$, and $O$ symbols. For instance, \refeq{eq:result} implies that there exists $\alpha, n_0>0$ such that $E_{t}(\Psi)  >\alpha n$ for all $n>n_0$.}
		the multiplicative coefficient depends only on $t$.
	\end{result}
	
	The proof involves several steps, but the underlying physical picture is simple enough to summarize here.  
	Recall that the toric code supports two types of Abelian anyonic excitations, $e$ and $m$, that can be created at end points of stringlike operators. 
	Such operators can be deformed while fixing the end points and create the same sets of anyons. 
	The braiding statistics between $e$ and $m$ can be probed as follows:
	\begin{align}\label{eq:introbraid}
		\langle \Psi| \gamma_m^{\dagger} \gamma_e \gamma_m | \Psi\rangle = \figbox{0.25}{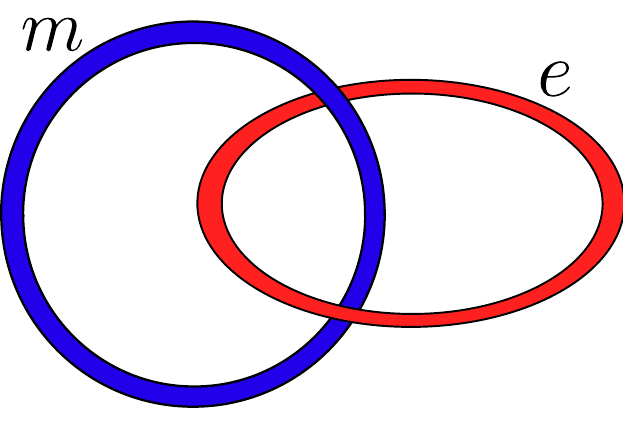}  = -1.
	\end{align}
	Here $\gamma_m$ represents an open $m$ string, and $\gamma_e$ represents a closed $e$ loop; in the middle is a spacetime diagram illustrating the braiding. 
	The key observation is that due to the deformability of $\gamma_m$ string, the region that supports the pair of $m$ anyons must be decoupled from the rest of the system if $\ket{\Psi}$ was short-range entangled, so that the $\gamma_e$ loop ``does not know" whether there was or was not an $m$ anyon inside. Therefore, $\langle \Psi | \gamma_m^{\dagger} \gamma_e \gamma_m | \Psi \rangle =1$.
	This contradicts the braiding statistics, which proves that the toric code ground states must be long-range-entangled. 
	The above beautiful argument is due to Bravyi (unpublished) and was described in Ref.~\cite{aharonov2018quantum}.
	
	To prove the claimed linear lower bound on $E_{t}(\Psi)$, we consider $O(n)$ copies of braiding processes in nonoverlapping subregions where each process contributes $O(1)$ to $E_{t}(\Psi)$. 
	This argument can be readily extended to ground states of any spin systems that support nontrivial anyonic excitations. 
	Namely, we can show that the nontrivial braiding $\langle \gamma_a^{\dagger} \gamma_b \gamma_a \rangle \not = 1$ for two anyons $a$ and $b$ is inconsistent with short-range entanglement, and as such, a similar lower bound on $E_{t}(\Psi)$ follows. 
	
	\subsection{Emergent fermions}
	
	Next, we turn our attention to systems that support emergent fermions. 
	To be concrete, we focus on the following variant of the Kitaev honeycomb model~\cite{kitaev2006anyons} (similar results can be derived in more general settings):
	\begin{align}\label{eq-Hfermion6}
		H_{\text{fermion}} = - \sum_{\hexagon} S_{\hexagon}, \qquad S_{\hexagon} = \figbox{0.25}{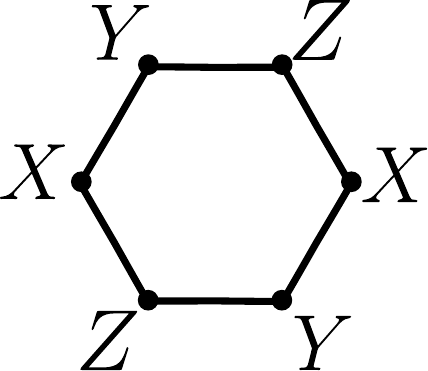}\ .
	\end{align}
	This Hamiltonian naturally emerges when viewing Kitaev's honeycomb model as a subsystem code and then considering the stabilizer group \cite{suchara2011constructions,ellison2023pauli}. 
	Viewing \refeq{eq-Hfermion6} as a stabilizer code, a ground state satisfies $S_{\hexagon}|\Psi\rangle = + |\Psi\rangle$, and the code parameters are $k \sim n/2$ and $d=2$. 
	Despite the exponentially large code subspace $2^k$ and the small code distance $d$, a linear lower bound, which is applicable to \emph{arbitrary} ground state, can be proven. 
	
	\begin{result}\label{res:fermion}
		A ground state $|\Psi\rangle$ of the honeycomb model satisfies 
		\begin{align}
			E_{t}(\Psi) =\Omega(n),
		\end{align}
		where the multiplicative coefficient depends only on $t$.
	\end{result}
	
	Unlike anyons in the toric code, fermions have trivial braiding statistics (i.e., they are not anyons). Hence, our proof utilizes the fermion exchange statistic crucially. 
	Before proceeding, however, it is worth emphasizing an important difference between how anyons and fermions emerge in the toric code and the honeycomb model, respectively. 
	In the toric code, open string anyon creation operators do not commute with the Hamiltonian; hence, anyons are low-energy excitations \emph{outside} the ground state subspace. 
	On the contrary, in the honeycomb model, the fermion creation operators do commute with the Hamiltonian, and, hence, fermions emerge \emph{inside} the ground state subspace. 
	More precisely, note the honeycomb model possesses the following two-body logical operators
	\begin{align}\label{eq:XXYYZZ}
		\figbox{0.3}{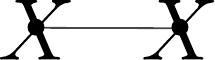}\qquad
		\figbox{0.3}{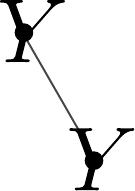}\qquad
		\figbox{0.3}{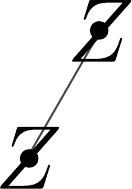}
	\end{align}
	that commute with stabilizers $S_{\hexagon}$. 
	Let us pick an arbitrary ground state $|\Psi\rangle$ and regard it as a ``vacuum'' state with no fermions. 
	Consider an open string operator $M$ which is generated from two-body logical operators. 
	Then, $M|\Psi\rangle$ can be interpreted as a state with two emergent fermions at end points of an open string $M$. 
	The crucial point is that $M|\Psi\rangle$ still lives in the energy ground space, as $M$ is a logical operator of the code.

	The particle exchange statistics can be probed by considering three stringlike operators $M_1, M_2$ and $M_3$ that share a common end point \cite{levin2003fermions}:
	\begin{eqnarray}
		\figbox{0.3}{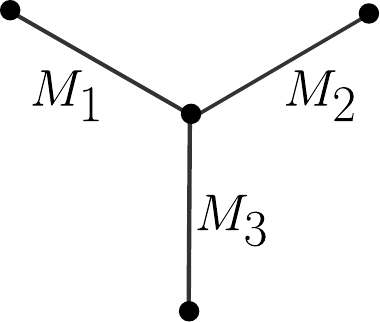} \qquad  
		\bra{\Psi}M_3^\dagger M_2^\dagger M_1^\dagger M_3M_2M_1\ket{\Psi}= e^{i \theta},
	\end{eqnarray}
	where $e^{i \theta}=-1$ for fermions. 
	In the honeycomb model, this relation can be explicitly verified, as three stringlike logical operators anticommute with each other at the overlapping qubit. 
	Using a variant of Bravyi's idea, it can be shown that short-range entanglement and the deformability of stringlike operators would imply $\bra{\Psi}M_3^{\dagger} M_2^{\dagger}M_{1}^{\dagger} M_3 M_2 M_1\ket{\Psi} = 1$, contradicting with the exchange statistics. 
	Then, with some extra work, the claimed linear lower bound on $E_{t}(\Psi)$ can be obtained. 
	
	\subsection{Anomalous symmetry and Mixed state}
	
	The aforementioned result on emergent fermions has interesting implications from the perspectives of anomalous symmetries and mixed states.

	\subsubsection{Anomalous symmetries}
	Recently, there has been considerable progress in understanding the relation between quantum entanglement and anomalous symmetries. 
	To motivate this viewpoint, consider a system with conventional symmetries, such as a $\mathbb{Z}_2$ global symmetry generated by $S = \otimes_{j=1}^nX_{j}$.
	The system is equipped with well-defined local $\mathbb{Z}_2$ charge generators which are simply $X_{j}$'s. 
	The symmetry constraint of the form $S|\Psi\rangle=|\Psi\rangle$ then dictates that the total $\mathbb{Z}_2$ charge in $|\Psi\rangle$, expressed as a product of local $\mathbb{Z}_2$ generators $X_j$, must be trivial ($+1$). 
	Furthermore, the system has a natural ``vacuum'' state, namely $|+\rangle^{\otimes n}$ for which local generators $X_j$ act trivially, and, thus, is not entangled.  
	
	On the contrary, anomalous symmetries do not admit locally definable charge generators. 
	A prototypical example of this phenomenon is a 1D $\mathbb{Z}_2$ anomalous global symmetry, often called the C$ZX$ symmetry $S_{\text{C}ZX}$~\cite{chen20122d}, which is generated by 
	\begin{align}
		S_{\text{C}ZX} = S_{\text{C}Z}S_{X} ,
	\end{align}
	where $S_{\text{C}Z} = \prod_j \text{C}Z_{j,j+1}, \; S_{X} = \prod_j X_{j}$. This operator defines a global $\mathbb{Z}_2$ symmetry as $S_{\text{C}ZX}^2 = I$ (for even $n$) but does not possess local generators. 
	For instance, one might want to consider $\text{C}Z_{j,j+1} X_{j}$ as a local generator. 
	Unfortunately, a product of $\text{C}Z_{j,j+1} X_{j}$ does not generate $S_{\text{C}ZX}= S_{\text{C}Z}S_{X}$. 
	Furthermore, $\text{C}Z_{j,j+1} X_{j}$ does not commute with $S_{\text{C}ZX}$, and also they do not commute with each other. 
	As such, consistent meaning as local charge generators cannot be given to them. 
	
	This impossibility of decomposing an anomalous symmetry operator into local charge generators is closely related to quantum entanglement. 
	This viewpoint has been extensively studied from the perspective of symmetry-protected topological (SPT) order, as the C$ZX$ symmetry emerges at the boundary of a 2D $\mathbb{Z}_{2}$ SPT phase~\cite{chen20122d}. 
	Recently, it has been rigorously shown that any wave functions satisfying $S_{\text{C}ZX}|\Psi\rangle = |\Psi\rangle$ must be multipartite entangled \cite{lessa2024mixed}.
	
	Returning to the honeycomb model, the upshot is that the hexagonal stabilizers $S_{\hexagon}$ are another example of anomalous symmetries.
	More precisely, these are $1$-form anomalous symmetries \cite{verresen2022unifying,ellison2023pauli,inamura2023fusion,liu2024symmetries} since $S_{\hexagon}$ acts as a codimension-$1$ operator instead of a global codimension-$0$ operator as in the C$ZX$ model. 
	The anomalous nature of $S_{\hexagon}$ symmetries can be seen from the fact that they cannot be cut into two pieces. For example, suppose we cut $S_{\hexagon}$ vertically into two parts and write it as a product of the left and right parts, $S_{\hexagon} = S_{\hexagon}^{(L)}S_{\hexagon}^{(R)}$.
	We immediately notice that $S_{\hexagon}^{(L)}$ and $S_{\hexagon}^{(R)}$ do not commute with some of other hexagonal stabilizers, and, thus, are not well defined as local generators.
	Instead, one might decompose $S_{\hexagon}$ as a product of two-body logical operators from Eq.~(\ref{eq:XXYYZZ}).
	While this decomposition consists of operators that commute with $S_{\hexagon}$ symmetries, two-body operators anticommute with each other, and, hence, do not admit interpretations as local generators.
	
	Our result on the honeycomb model can be rephrased as follows.
	\begin{result_restate}{res:fermion}{$'$}[Anomalous symmetry]
		Any state $|\Psi\rangle$ satisfying anomalous $1$-form symmetries $ S_{\hexagon}|\Psi\rangle=|\Psi\rangle$ in the honeycomb model must be long-range-entangled. Furthermore they satisfies $E_{t}(\Psi) =\Omega(n)$.
	\end{result_restate}

	While wave functions with anomalous global symmetries are long-range-entangled in general, they may score very low in terms of the GEM. For instance, the GHZ state $\frac{1}{\sqrt{2}}(|0\rangle^{\otimes n} + |1\rangle^{\otimes n}) $ satisfies the C$ZX$ symmetry, but it has $E_0(\text{GHZ})=1$.
	This is in contrast with anomalous $1$-form symmetries, such as $S_{\hexagon}$ symmetries, where the GEM scales linearly with respect to $n$.

	\subsubsection{Mixed states}
	Characterization of quantum entanglement for mixed states is a considerably harder question compared to its counterpart for pure states. 
	Quantum circuit complexity of preparing a mixed state has been studied in quantum information theory in the context of topological quantum memory at finite temperature~\cite{dennis2002topological,hastings2011topological}. 
	Recently, there has been renewed interest in the possibility of novel quantum phases, that are genuinely intrinsic to mixed state, in the condensed matter community.  
	Generally, a mixed state $\rho$ is said to be short-range-entangled if $\rho$ can be expressed as a probabilistic ensemble of short-range-entangled states~\cite{hastings2011topological}:
	\begin{align}
		\rho = \sum_{j}p_{j}|\psi_j\rangle \langle \psi_j|, \qquad |\psi_j\rangle = U_j |0\rangle^{n},
	\end{align}
	where $U_j$ are short-depth circuits.  
	A naturally arising question is whether a mixed state $\rho$ in the ground state space of the honeycomb model is long-range-entangled or not. 
	Evidence for long-range entanglement in $\rho$ was reported in Ref.~\cite{wang2023intrinsic}, where a subleading contribution in the entanglement negativity \cite{lu2020detecting}, similar to topological entanglement entropy, was found.
	Further connections between anomalous 1-form symmetry and mixed-state topological order were discussed in Refs.~\cite{sohal2024noisy,ellison2024towards}.
	Here, as a simple corollary of the aforementioned result, we can prove that $\rho$ in the honeycomb model is indeed long-range-entangled. 
	
	\begin{result_restate}{res:fermion}{$''$}[Mixed state]
		Any mixed state $\rho$ satisfying the hexagonal constraints $\Tr( \rho S_{\hexagon})=1$ in the honeycomb model cannot be written as an ensemble of short-range-entangled states. 
	\end{result_restate}
	
	In fact, we prove a stronger statement by generalizing the depth-$t$ GEM to the mixed-state setting where the maximum fidelity between the target state and all short-range-entangled mixed states is considered. 
	Namely, we prove that a similar linear lower bound on the mixed-state GEM holds for any symmetric $\rho$. 
	
	Lastly, we note that the results also apply to other models of mixed states with anomalous 1-form symmetry, such as the $ZX$-dephased toric code \cite{ellison2023pauli}.

	\subsection{2D stabilizer codes}
	
	The phenomenon of emergent anyons can be understood within the framework of quantum error correction (QEC) where quantum information is stored in the energy ground state space. 
	This motivates us to ask whether there is a fundamental lower bound on the necessary quantum entanglement for storing quantum information in a bounded volume of space. 
	
	In this paper, we reveal a universal lower bound on $E_{t}(\Psi)$ for certain two-dimensional QEC systems. 
	Specifically, we consider a system of $n$ qubits placed on a 2D lattice and focus on codes for which the code subspace $\mathcal{C}$ is defined by geometrically local stabilizer generators $S_{j}$: The code subspace $\mathcal{C}$ can be realized as the energy ground subspace of a local gapped Hamiltonian 
	\begin{align}
		H = - \sum_{j} S_{j}.
	\end{align}
	The QEC capability is usually measured by its code distance. A code is said to have code distance $d$ if logical states are locally indistinguishable on any subset of less than $d$ physical qubits. 
	
	Our main result is the following lower bound on the GEM. 
	
	\begin{result}\label{result:2dstabcode}
		For a 2D geometrically local stabilizer code, if the code distance $d>d_0$, then any logical state $|\Psi\rangle$ satisfies 
		\begin{align}
			E_{t}(\Psi) =\Omega(d^2).
		\end{align}
		Here, both the multiplicative constant and $d_0$ depend only on $t$ and the size of stabilizer generators. 
	\end{result}
	
	For the toric code, $d=\Theta(\sqrt{n})$; hence, Result \ref{result:2dstabcode} is in accordance with Result \ref{result:toric}. There is, however, a crucial difference.
	The starting point of Result \ref{result:toric} is the emergent anyons, which exist on any geometries including the sphere. 
	The starting point of Result \ref{result:2dstabcode}, on the other hand, is the QEC property which relies on the ground state degeneracy.
	Therefore, Result \ref{result:2dstabcode} does not apply to the toric code on the sphere \emph{per se}. 
	
	Nevertheless, the proof of result \ref{result:2dstabcode} reflects a rigorous and direct connection between QEC and emergent anyons. 
	The main idea is that we can derive the existence of emergent anyons \emph{solely} based on the error correction property. 
	More precisely, as long as the code distance $d$ is sufficiently larger than the size of local stabilizer generators, we can always construct looplike operators $\gamma_b$ that create only syndrome at end points, as well as stringlike stabilizer operators $\gamma_a$, such that $\braket{\gamma_a^\dagger\gamma_b\gamma_a}=-1$. 
	Key technical tools in the constructions of loop and stringlike operators are the cleaning lemma \cite{bravyi2009no} and the disentangling lemma \cite{bravyi2010tradeoffs} which allows us to identify a pair of anticommuting logical operators. 
	Specifically, these logical operators can be designed to be supported on unions of linelike subregions. 
	We then introduce a simple but powerful truncation procedure of stabilizer operators which enables us to obtain the braiding configuration similar to Eq.~(\ref{eq:introbraid}).
	Such braiding statistics is not compatible with short-range entanglement, which ultimately leads to the desired linear lower bound on the GEM.
	
	The anyon braiding statistics then simply follow by interpreting the syndromes at the end of string operators as anyons.
	From this point of view, our result implies that anyons always emerge from QEC in 2D stabilizer codes, namely, from the sole assumption of the code distance. 
	
	\subsection{Relation to previous work}
	
	In a previous work, we derived generic lower bounds on depth-$t$ GEM for QEC codes in several settings~\cite{bravyi2024much}. 
	The most relevant to the present paper is the following universal lower bound for a quantum low-density parity-check (LDPC) code with the code distance $d$: 
	\begin{align}
		E_{t}(\Psi) =\Omega(d),
	\end{align}
	where the multiplicative coefficient depends only on the sparsity of the qLDPC code and the depth $t$.
	The linear scaling with $d$ cannot be improved further, as there exist qLDPC codes with the distance $d\sim n$. 
	This lower bound, however, is not tight for quantum codes supported on lattices with geometrically local projectors. 
	For instance, the two-dimensional toric code has $E_{0}(\Psi)\sim n/2$, but its code distance is $d \sim \sqrt{n}$. 
	Also, the above lower bound is not applicable to the toric ``code'' without logical qubits (e.g., supported on a sphere) as its derivation crucially depends on the presence of logical qubits. 
	Finally, the above bound is not particularly useful in studying long-range entanglement in the honeycomb model as its code distance is $d=2$.
	
	Our results in the present work improve the above qLDPC lower bound in various ways by focusing on two-dimensional systems. 
	Result~\ref{result:toric} proves that the quadratic scaling $E_t(\text{toric})=\Omega(d^2)$ [$t\sim O(1)$] for the toric code. 
	Result~\ref{result:2dstabcode} further improves this characterization for the cases of 2D local stabilizer codes by proving $E_{t}(\Psi)=\Omega(d^2)$. 
	Result~\ref{res:fermion} proves that, despite the poor coding properties, all the code word states in the honeycomb model are long-range-entangled with $E_{t}(\Psi) =\Omega(n)$.
	
	\subsection{Plan of the paper}
	
	This paper is organized as follows.
	In Sec.~\ref{sec:toric}, we study the toric code as a prototypical example of systems with emergent anyons.
	In Sec.~\ref{sec:storage}, we will derive a generic lower bound on the GEM for arbitrary 2D local stabilizer codes. 
	In Sec.~\ref{sec:fermion}, we will study the honeycomb model and show that wave functions with anomalous $1$-form symmetries have $E_t{(\Psi)}=\Omega(n)$.  
	In Sec.~\ref{sec:discussion}, we conclude with brief discussions.
	
	\section{2D toric code}\label{sec:toric}
	
	We prove the following linear lower bound on  $E_{t}(\Psi)$ for a toric code ground state.
	
	\begin{theorem}\label{thm:toriccode}
		There exists an absolute constant $\alpha>0$ ($\alpha=10^{-4}$ suffices) such that for any toric code ground state $|\Psi \rangle$, if $t\leq O(\sqrt{n})$, then
		\begin{equation}
			E_{t}(\Psi) > \frac{\alpha n}{(t+1)^2} \ .
		\end{equation}
	\end{theorem}

	\label{sec:braiding}
	\subsection{Anyon braiding statistics}
	
	Recall that the toric code Hamiltonian is given by 
	\begin{align}
		H_{\text{toric}} = - \sum_{+}A_{+}- \sum_{\square}B_{\square}, 
	\end{align}
	where $A_{+} = \prod_{j\in+}X_{j}$ and $  B_{\square} = \prod_{j\in \square}Z_{j}$.  Qubits reside on edge of an $L\times L$ square lattice, and $+$ ($\square$) denote vertices (plaquettes). 
	A ground state $|\Psi\rangle$ satisfies $A_{+}|\Psi\rangle=B_{\square}|\Psi\rangle=|\Psi\rangle$.
	
	Microscopic details of the construction, however, are not important. 
	A crucial property we use for the proof is the nontrivial braiding statistics of anyons of the toric code. 
	Namely, consider a pair of a closed string operator $\gamma_e$ and an open string operator $\gamma_m$ which consist of Pauli $Z$ and $X$ operators respectively (Fig.~\ref{fig:dressed}). Considering a process of braiding $e$ around $m$, we find
	\begin{align}
		\bra{\Psi}\gamma_{m}^\dagger\gamma_{e}\gamma_{m}|\Psi\rangle = -1. \label{eq:braiding}
	\end{align}
	This braiding statistics is stable under deformation by a constant depth circuit $U_t$.
	Namely, for a deformed toric code ground state $|\Psi_t\rangle \equiv U_t|\Psi\rangle$, define dressed string operators by $\widetilde{\gamma_m} = U_t \gamma_m U^{\dagger}_t$ and $\widetilde{\gamma_e} = U_t \gamma_e U^{\dagger}_t$. 
	Equation (\ref{eq:braiding}) still holds for $|\Psi_t\rangle$ and  $\widetilde{\gamma_m}, \widetilde{\gamma_e}$, but dressed string operators $\widetilde{\gamma_m}, \widetilde{\gamma_e}$ now have width $O(t)$.  
	One can then consider a quasilocal anyon braiding process as in Fig.~\ref{fig:dressed} by taking a larger region of linear size $t$. 
	
	\subsection{Long-range entanglement}\label{sec:TCLRE}
	
	We first prove that a toric code ground state is long-range-entangled. 
	
	\begin{lemma}\label{lemma:TCLRE}
		Given a toric code ground state $|\Psi\rangle$, we have
		\begin{align}
			U_t|\Psi\rangle \not=  \ket{0}^{\otimes n} \label{eq:step1}
		\end{align}
		for any depth-$t$ unitary $U_t$ with constant $t$.
	\end{lemma}
	
	\begin{figure}[h!]
		\centering 
		\includegraphics[width = 6cm, height = 5.7cm]{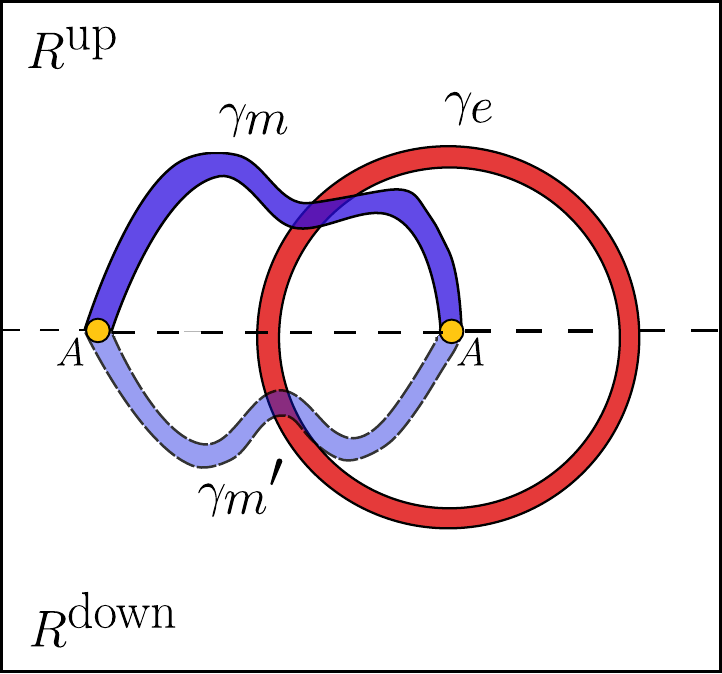}
		\caption{Two dressed anticommuting string operators for anyon braiding in a deformed toric code. Their width is at most $2(t+1)$. For the purpose of the proof, it suffices to consider a region $R$ with a linear size of $8(t+1)$. }
		\label{fig:dressed}
	\end{figure}
	
	\begin{proof}
		The proof is similar to Bravyi's proof described in Ref.~\cite{aharonov2018quantum}. Suppose that Eq.~(\ref{eq:step1}) is false. 
		In terms of a deformed state $|\Psi_{t}\rangle \equiv U_t |\Psi\rangle $, this is equivalent to $|\Psi_t\rangle = \ket{0}^{\otimes n}$. Since $U_t$ is a finite depth circuit, the braiding process of \refeq{eq:braiding} can be considered by a pair of dressed string operators $\widetilde{\gamma_m}=U_t\gamma_m U_t^{\dagger}$ and $\widetilde{\gamma_e}=U_t\gamma_e U_t^{\dagger}$. 
		Below, we omit the tilde and denote dressed operators by $\gamma_m,\gamma_e$ to simplify the notation.
		
		As in Fig.~\ref{fig:dressed}, we use a dressed string operator $\gamma_m$ to create two $m$ anyons supported on region $A$. We claim that
		\begin{equation}\label{eq:step1decouple}
			\gamma_m\ket{\Psi_t}=\ket{\phi_A}\otimes\ket{0_{A^c}} \ ,
		\end{equation}
		where $\ket{\phi_A}$ is a state in the Hilbert space $\mathcal{H}_A$ and $\ket{0_{A^c}}$ is the product state of $\ket{0}$ for all qubits outside $A$.
		To show this, we divide the system into three subregions: $A$, $R^{\text{up}}$, and $R^{\text{down}}$ as in Fig. \ref{fig:dressed}. 
		We can find another string operator $\gamma_m'$ that is supported on $A\cup R^{\text{down}}$ satisfying
		\begin{equation}\label{eq-deformability}
			\gamma_m'\ket{\Psi_t}=\gamma_m\ket{\Psi_t}.
		\end{equation}
		Define $\Pi^{\text{up}}$ as an operator that projects $\mathcal{H}_{R^{\text{up}}}$ to $\ket{0_{R^{\text{up}}}}$. We have, by assuming Eq.(\ref{eq:step1}) is false,
		\begin{equation}    
			\begin{aligned}
				\Pi^{\text{up}}\gamma_m'\ket{\Psi_t}
				=&\gamma_m'\Pi^{\text{up}}\ket{\Psi_t}
				=\gamma_m'\Pi^{\text{up}}\ket{0}^{\otimes n}\\
				=&\gamma_m'\ket{0}^{\otimes n}
				=\gamma_m'\ket{\Psi_t}.
			\end{aligned}
		\end{equation}
		Similarly, 
		\begin{equation}
			\Pi^{\text{down}}\gamma_m\ket{\Psi_t}=\gamma_m\ket{\Psi_t}.
		\end{equation}
		Combining the above three equations, we have
		\begin{equation}
			\Pi^{\text{up}}\Pi^{\text{down}}\gamma_m\ket{\Psi_t}=\gamma_m\ket{\Psi_t},
		\end{equation}
		and, hence, Eq.~(\ref{eq:step1decouple}) is proved.
		
		Now we apply a loop operator $\gamma_e$ on $\gamma_m\ket{\Psi}$. Recall that $\gamma_e$ is a deformed stabilizer, and, hence, $\gamma_e|\Psi_t\rangle = |\Psi_t\rangle$.
		Assuming Eq.(\ref{eq:step1}) is false, this implies $\gamma_e\ket{0_{A^c}}=\ket{0_{A^c}}$.
		Therefore, due to Eq.~(\ref{eq:step1decouple}), we have
		\begin{equation}
			\begin{aligned}
				\gamma_e\gamma_m\ket{\Psi_t}
				=&\ket{\phi_A}\otimes\gamma_e\ket{0_{A^c}}
				=\ket{\phi_A}\otimes\ket{0_{A^c}}\\
				=&\gamma_m\ket{\Psi_t}.
			\end{aligned}
		\end{equation}
		This, however, contradicts \refeq{eq:braiding}, namely $\gamma_m\gamma_e\gamma_m|\Psi_t\rangle = - |\Psi_t\rangle$, which completes the proof.
	\end{proof}
	
	It is worth mentioning that the proof based on the Lieb-Robinson bound in Ref.~\cite{bravyi2006lieb} crucially uses the ground state degeneracy, while the above proof (based on Ref.~\cite{aharonov2018quantum}) applies to the toric code with no ground state degeneracy as well.
	
	\subsection{Overlap with short-range-entangled state}
	
	The result of the above lemma is equivalent to $E_{t}(\Psi)>0$, and is, in general, not sufficient to obtain a lower bound for $E_{t}(\Psi)$ that would scale with the system size. 
	
	We now prove an upper bound on the maximum overlap between $\rho_R$ of a deformed toric code ground state and $|0_R\rangle$ for a patch $R$ of linear size $~ t$. 
	
	\begin{lemma}\label{lemma:step2toric}
		Let $R$ be any patch of linear size $8(t+1)$. 
		Given a deformed toric code ground state $|\Psi_t \rangle = U_t |\Psi\rangle$  with a depth $t$ circuit $U_t$, 
		there exists an absolute constant $\epsilon >0$ such that, in terms of trace norm,   
		\begin{equation}\label{eq:toricbound}
			\big\lVert\rho_R-\ketbra{0_R}\big\rVert_{1}>\epsilon,
		\end{equation}
		where $\rho_{R} = \Tr_{R^c}(|\Psi_t \rangle \langle \Psi_t|)$, and accordingly the fidelity\footnote{
			Here we use the convention $\mathcal{F}(\rho,\sigma) = \Big(\Tr \sqrt{ \sqrt{\rho} \sigma \sqrt{\rho}}\ \Big)^2$.
		} satisfies: 
		\begin{eqnarray}
			{\cal F}(\ket{0_R}, \rho_{R}) < 1 - \frac{1}{4}\epsilon^2=1-\epsilon'
		\end{eqnarray}
		for an absolute constant $\epsilon' > 0$.
	\end{lemma}
	To rephrase, the ``local GEM" is not only positive, but must also bounded away from zero by an absolute gap that is independent of the subregion size.

	\begin{proof}
		Let $\gamma_m$ and $\gamma_e$ be dressed string operators. 
		Suppose that Eq.~(\ref{eq:toricbound}) is false. We claim that
		\begin{equation}\label{eq:puredecomposition}
			\gamma_m |0_R \rangle \langle 0_R| \gamma^\dagger_m
			\approx \ketbra{0_{A^c}} \otimes \rho'_A,
		\end{equation}
		where $\approx$ means that the trace distance is smaller than a constant $O(\epsilon)$. 
		Here, $A^c$ refers to the complement of $A$ inside the patch $R$; $\rho'_A$ is positive semidefinite but might be unnormalized, since we allow approximation.
		
		The proof is a more quantitative version of Lemma \ref{lemma:TCLRE}. We have
		\begin{eqnarray}\label{eq:4episilon}
			&&\Pi^{\text{up}} \big(\gamma_m |0_R \rangle \langle 0_R| \gamma^\dagger_m \big)\Pi^{\text{up}}  \cr
			&\approx& 
			\Pi^{\text{up}} \big(\gamma_m \rho_{R} \gamma^\dagger_m \big) \Pi^{\text{up}} \cr 
			&=&  \Pi^{\text{up}}
			\big(\gamma_{m'} \rho_{R} \gamma^\dagger_{m'} \big) \Pi^{\text{up}} 
			= \gamma_{m'} \big(\Pi^{\text{up}} \rho_{R}  \Pi^{\text{up}} \big) \gamma^\dagger_{m'} \cr 
			&\approx&  \gamma_{m'} \big(\Pi^{\text{up}} \ket{0_R}\bra{0_R}  \Pi^{\text{up}} \big) \gamma^\dagger_{m'}
			=\gamma_{m'} \big(\ket{0_R}\bra{0_R}  \big) \gamma^\dagger_{m'} \cr 
			&\approx& \gamma_{m'} \rho_{R} \gamma^\dagger_{m'} = \gamma_m \rho_{R} \gamma^\dagger_{m} \cr 
			&\approx& \gamma_m |0_R \rangle \langle 0_R| \gamma^\dagger_m,
		\end{eqnarray}
		where the negation of Eq.~(\ref{eq:toricbound}) is used for each $\approx$.
		Moreover,
		\begin{align}
			&\Pi^{\text{down}} (\gamma_m |0_R \rangle \langle 0_R| \gamma^\dagger_m ) \Pi^{\text{down}} \\
			=\;&\gamma_m \Pi^{\text{down}} |0_R \rangle \langle 
			0_R | \Pi^{\text{down}}\gamma^\dagger_m 
			= \gamma_m |0_R \rangle \langle 0_R|
			\gamma^\dagger_m. \notag
		\end{align}
		Therefore, \refeq{eq:puredecomposition} is proved:
		\begin{eqnarray}\label{eq-30}
			\gamma_m |0_R \rangle \langle 0_R| \gamma^\dagger_m  &\approx&\Pi^{\text{down}}\Pi^{\text{up}} \big(\gamma_m |0_R \rangle \langle 0_R| \gamma^\dagger_m \big) \Pi^{\text{up}} \Pi^{\text{down}} \cr
			&=& 
			\ketbra{0_{A^c}}\otimes \rho'_A.
		\end{eqnarray}
		
		Using this approximation, we obtain 
		\begin{eqnarray}\label{eq:31}
			&&\Tr(\gamma_e \gamma_m \rho_{R} \gamma^\dagger_m) \cr
			&\approx& \Tr(\gamma_e \gamma_m  |0_R \rangle \langle 0_R | \gamma^\dagger_m) \cr 
			&\approx& \Tr[ (\gamma_e \ketbra{0_{A^c}}) \otimes \rho'_A]  
			= \Tr(\gamma_e \ketbra{0_{A^c}}) \Tr(\rho'_A) \cr
			&\approx& \Tr(\gamma_e \ketbra{0_{A^c}})
			= \Tr(\gamma_e |0_R \rangle \langle 0_R|)\cr 
			&\approx& \Tr (\gamma_e \rho_{R}) = \bra{\Psi_{t}}\gamma_e\ket{\Psi_{t}}= 1
		\end{eqnarray}
		which contradicts \refeq{eq:braiding}.
		
		Counting the $\epsilon$'s in the above calculation, we see $\epsilon=\frac{1}{5}$ suffices\footnote{More precisely, \refeq{eq:4episilon} implies that the 1-norm distance between two sides is less than $4\epsilon$; the approximation in \refeq{eq-30} is therefore accurate up to $4\epsilon$; the approximations in \refeq{eq:31} are accurate up to $\epsilon$, $4\epsilon$, $4\epsilon$, $\epsilon$, respectively, leading to a total difference up to $10\epsilon$. However, the left-hand side of \refeq{eq:31} should equal -1 assuming \refeq{eq:braiding}.}; hence, $\epsilon'=0.01$ suffices.
		Importantly, $\epsilon$ and $\epsilon'$ are absolute constants and do not depend on the size of $R$. 
	\end{proof}
	
	We emphasize that the proof applies whenever anyons can be excited and braiding can be performed. 
	It does not depend, for example, on the microscopic details of the stabilizers, the specifics of the lattice, or the translational invariance of the system.

	Finally, to complete the proof of Theorem~\ref{thm:toriccode}, we divide the whole system into multiple regions $R_i$ such that mutual distances are larger than $2(t+1)$.
	For a (deformed) toric code, we have the following decoupling property: 
	\begin{equation}
		\rho_{R}=\underset{\scriptscriptstyle i}\otimes \rho_{R_i},\qquad R=\cup_i R_i,
		\label{eq:decoupling}
	\end{equation}
	since stabilizer generators are geometrically local and each region $R_i$ is correctable~\cite{bravyi2010tradeoffs}. From this and the monotonicity of fidelity, we obtain 
	\begin{equation}\label{eq:33} 
		\begin{aligned}
			{\cal F}(\ket{\Psi_t},\ket{0^n}) 
			&\leq {\cal F} (\rho_R, \ket{0_R}) 
			= \prod_i {\cal F} (\rho_{R_i}, \ket{0_{R_i}}) \\
			&< (1-\epsilon')^{\frac{n}{100(t+1)^2}},
		\end{aligned}
	\end{equation}
	where $n/100(t+1)^2$ is a number of disentangled patches. 
	We emphasize that \refeq{eq:decoupling} is essential for \refeq{eq:33}, otherwise the first inequality does not hold, with the GHZ state being a counterexample.
	Equation \ref(eq:33) proves Theorem~\ref{thm:toriccode}, with the constant $\alpha=-\frac{1}{100}\log_2(1-\epsilon')\approx 1.4\times 10^{-4}$.

	While we focused on the toric code for concreteness, our proof is applicable to any system featuring emergent anyons with nontrivial braiding statistics.
	The key ingredients for our argument are the braiding relation \refeq{eq:braiding} and the deformability of string operators \refeq{eq-deformability}.
	In more general settings, \refeq{eq:braiding} takes the value $\mathcal{D}S_{ab}/d_ad_b$, where $d_a$ and $d_b$ are the quantum dimension of the anyons, $\mathcal{D}$ is the total dimension, and $S_{ab}$ is the corresponding element of the modular $S$ matrix \cite{kitaev2006anyons}. 
	Our proof applies whenever $\mathcal{D}S_{ab}/d_ad_b\neq 1$, namely, whenever the braiding is nontrivial.

	\section{2D local stabilizer codes}\label{sec:storage}
	In the previous section, we derived the GEM lower bound for the toric code based on emergent anyons.
	The toric code model stands out as a quintessential example of topological quantum error-correcting codes.
	In this section, we show that a similar bound applies more broadly to all 2D geometrically local stabilizer codes with or without boundaries (the geometric locality here is defined using the standard Euclidean geometry).
	
	\begin{theorem}\label{thm_stabilizer}
		For 2D geometrically local stabilizer codes, denote $w$ as the maximal diameter of stabilizer generators. If $d>\Theta((w+t)^2)$, we have:
		\begin{equation}
			E_{t}(\Psi)>\alpha\frac{d^2}{w^2(w+t)^2}.
		\end{equation}
		for an absolute constant $\alpha>0$.
	\end{theorem}
	
	Although the starting point here is quantum error correction instead of anyons, our proof strategy is similar to the one for the toric code using nontrivial braiding processes. 
	The key extra ingredient is a generic prescription for constructing string operators.
	This enables us to derive the existence of anyons with nontrivial braiding statistics based solely on the quantum error correction property.

	\subsection{Cleaning lemma}
	
	A technical challenge is to construct a pair of open and closed string operators that anticommute with each other without relying on the macroscopic details of stabilizer generators.
	The crucial tool is a particular strengthening of the cleaning lemma for 2D local stabilizer codes~\cite{bravyi2009no}. 
	
	Consider square regions $A_{j}$ of linear size $O(\frac{d}{w})$ which are separated from each other by $O(w)$. 
	Denote the union by $A\equiv\cup_{j}A_{j}$, and its complement by $B\equiv A^c$. 
	We call a subsystem $B$ a \emph{mesh} (Fig.\ref{fig_mesh}).  
	
	\begin{figure}[h!]
		\centering
		\includegraphics[width=0.3\textwidth]{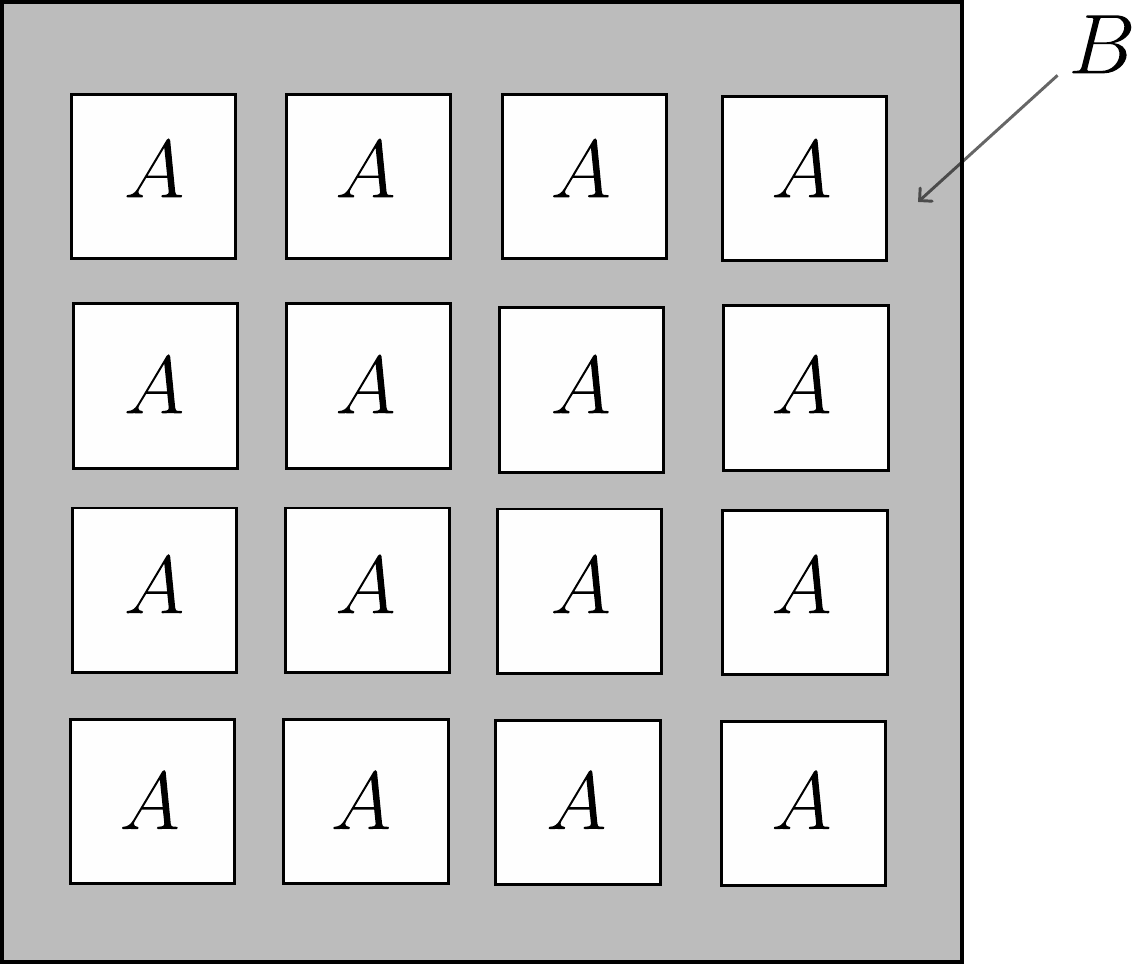}
		\caption{Each square $A_j$ is correctable and of linear size $O(\frac{d}{w})$. Squares are separated by $O(w)$, so their union $A$ is also correctable. Hence, all the logical operators can be found in the shaded region, a mesh $B=A^c$. }
		\label{fig_mesh}
	\end{figure}
	
	\begin{lemma}\label{lemma:cleanmesh}
		In a 2D local stabilizer code, all logical operators can be supported on a two-dimensional mesh $B=A^c$. 
	\end{lemma}
	
	\begin{proof}
		The lemma was essentially proven in Ref.~\cite{bravyi2010tradeoffs}, and we sketch the key arguments. 
		First, each square is correctable since the linear size is $O(\frac{d}{w})$ (Corollary 1 in Ref.~\cite{bravyi2010tradeoffs}).
		Second, since their mutual separation is larger than the size of stabilizers, the union $A=\cup_j A_j$ is also correctable (Lemma 2 in Ref.~\cite{bravyi2010tradeoffs}).
		Therefore, due to the cleaning lemma (Lemma 1 in Ref.~\cite{bravyi2009no}), for any logical operator $\ell$, there exists a stabilizer $S$ such that $\ell S$ is supported on the mesh $B=A^c$.
	\end{proof}
	
	We use Lemma~\ref{lemma:cleanmesh} to construct a braiding process. 
	Let $\ell_1,\ell_2$ be a pair of anticommuting logical operators.
	Picking a mesh $B_1$ as in Lemma \ref{lemma:cleanmesh}, $\ell_1$ can be supported on $B_1$.
	We now shift $B_1$ by $O(w)$ in both horizontal and vertical directions and denote the new mesh by $B_2$ (Fig.~\ref{fig_shift}).
	Then, $\ell_2$ can be supported on $B_2$.
	Observe that the intersection of $\ell_1$ and $\ell_2$ is a union of squares of linear size $O(w)$. 
	Since $\ell_1$ and $\ell_2$ anticommute with each other, they must anticommute on at least one square. Denote such a square as $Q$. 
	
	\begin{figure}[h!]
		\centering
		\includegraphics[width=0.45\textwidth]{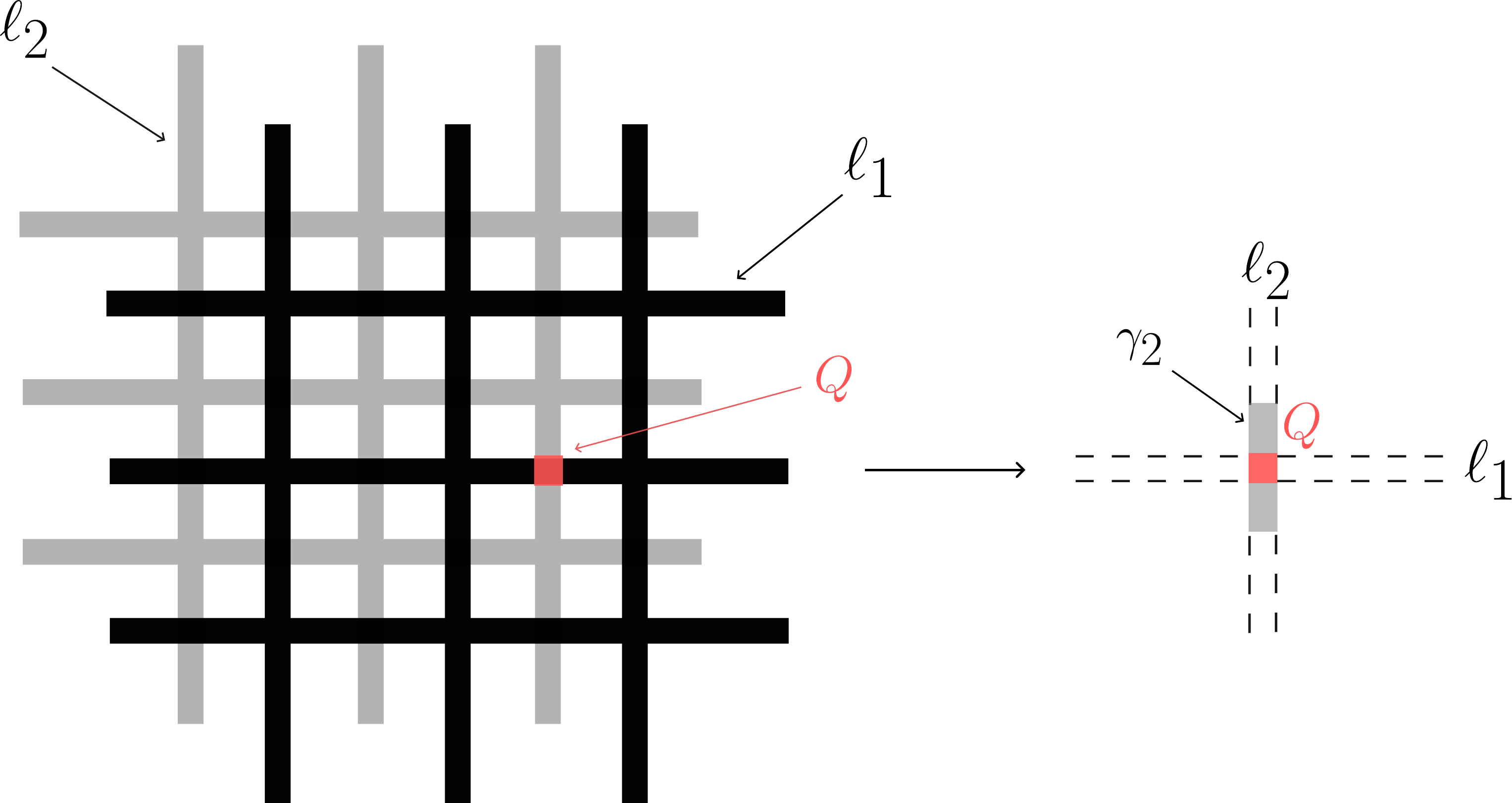}
		\caption{
			(Left) A pair of logical operators $\ell_1$ and $\ell_2$ which anticommute with each other at one square $Q$. 
			(Right) A segment of $\ell_2$ around $Q$ is called $\gamma_2$. 
		}
		\label{fig_shift}
	\end{figure}

	\subsection{Braiding process}
	
	The goal is to find a braiding process of the form $\gamma_2 \gamma_1 \gamma_2|\Psi\rangle = - |\Psi\rangle$ where $\gamma_1$ is a looplike stabilizer and $\gamma_2$ is a deformable open string operator which anticommutes with $\gamma_1$.  
	To construct $\gamma_1$, we repeat the above argument for a slightly deformed mesh $B_1'$ and construct $\ell_1'$, a logical Pauli operator that is equivalent to $\ell_1$ and is supported on $B_1'$. 
	This is possible as long as the complement of $B_{1}'$ consists of squares whose perimeter remains $O(\frac{d}{w})$. 
	Then, $\ell_1\ell_1'$ is a stabilizer operator supported on $B_1\cup B_1'$ (Fig.~\ref{fig_truncation}). 
	
	We construct a looplike stabilizer $\gamma_1$ from $\ell_1\ell_1'$ by ``truncating'' its support away from the intersection $Q$. 
	For a given region $R$, let $R^+$ denote the set of qubits $j$ such that $\disc(j,R)\leq w$.
	The following lemma allows us to truncate a large stabilizer $S$ into a smaller stabilizer $S'$ supported on $R^+$ with possible changes only at the boundary while keeping the operator content on $R$.

	\begin{figure}[h!]
		\centering
		\includegraphics[width=0.5\textwidth]{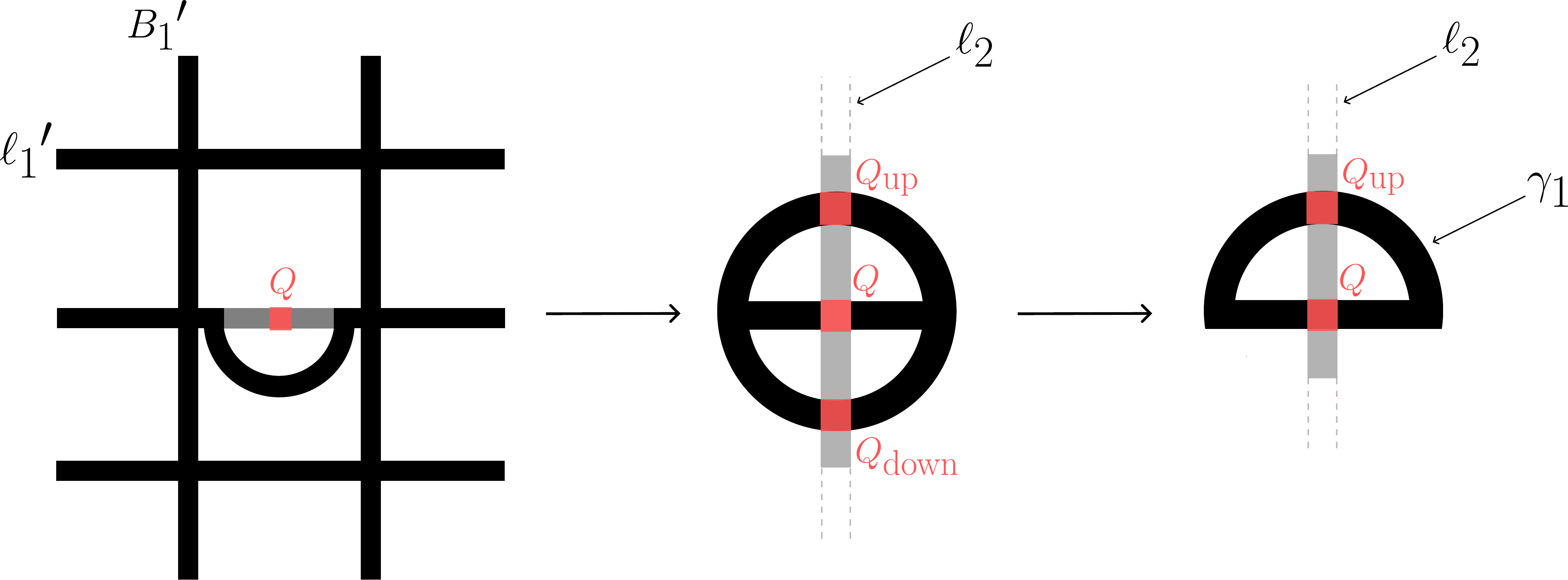}
		\caption{
			(Left) An equivalent logical operator $\ell_1'$ supported on a deformed mesh $B_{1}'$. 
			(Middle) Starting from a stabilizer operator $S=\ell_1\ell_1'$ supported on $B_1\cup B_{1}'$, we truncate $S$ to obtain a stabilizer operator $\gamma_1$ supported around $Q$. This stabilizer intersects with $\ell_2$ at three squares. 
			(Right) We further truncate the stabilizer to obtain $\gamma_1$ which intersects with $\ell_2$ at two squares. 
		}
		\label{fig_truncation}
	\end{figure}
	
	\begin{lemma}
		For any stabilizer operator $S$ and any region $R$, there exists a stabilizer operator $S'$, such that (i) $\supp(S')\subseteq R^+$ and (ii) $\supp(SS')\cap R=\emptyset$ (equivalently, $S|_{R} = S' |_{R}$).
	\end{lemma}
	\begin{proof}
		Since $S$ is a stabilizer operator, it can be written as a product of local stabilizer generators $S=\prod_{a}S_{a}$. 
		Specifically, we can decompose $S$ as follows:
		\begin{align}
			S = \prod_{\supp(S_a)\cap R=\emptyset} S_{a} \prod_{\supp(S_b)\cap R\not=\emptyset} S_{b}.
		\end{align}
		Define $S'=\prod_{\supp(S_b)\cap R\not=\emptyset} S_{b}$. 
		Since the diameter of each $S_a$ is less than $w$, we know $\supp(S')\subseteq R^+$.
		Moreover, $\supp(SS')\cap R=\emptyset$ since $SS'=\prod_{\supp(S_a)\cap R=\emptyset} S_{a}$.
	\end{proof}
	
	The construction of $\gamma_1$ proceeds in two steps. 
	First, let us apply the above lemma to truncate $\ell_1\ell_1'$ into a circular region containing the square $Q$. 
	The resulting stabilizer operator is supported on a region consisting of a circle and a horizontal line; see the middle panel of Fig.~\ref{fig_truncation} for illustration. 
	Second, we observe that this truncated stabilizer intersects with $\ell_2$ at $Q_{\text{up}}$, $Q$, and $Q_{\text{down}}$.
	Since this operator anticommutes with $\ell_2$ at $Q$, it must anticommutes with $\ell_2$ at either $Q_{\text{up}}$ or $Q_{\text{down}}$. 
	Without loss of generality, we assume that it is $Q_{\text{up}}$. 
	We can then further shrink the stabilizer operator to a semicircular region, as shown in the right panel of Fig.~\ref{fig_truncation}, by removing the stabilizer generators below the horizontal line.
	
	This is the desired looplike stabilizer which we shall use as $\gamma_1$. 
	Note that $\gamma_1$ intersects and anticommutes with $\ell_2$ at both $Q$ and $Q_{\text{up}}$. 
	
	The construction of $\gamma_2$ proceeds in an analogous manner. 
	We start by finding $\ell_2'$ in a deformed mesh $B_2'$, and then truncate $\ell_2 \ell_2'$ so that it will be supported only near $Q_{\text{up}}$. 
	We can then construct a looplike stabilizer operator $S_{2}$ which intersects and anticommutes with $\gamma_1$ at $Q_{\text{up}}$ and $Q_{\text{up}}'$; see Fig.~\ref{fig:braiding} for an illustration.
	Finally, we decompose $S_{2}$ into two segments, $S_{2}=\gamma_2 \gamma_2'$, such that $\gamma_2$ and $\gamma_2'$ intersect with $\gamma_1$ at $Q_{\text{up}}$ and $Q_{\text{up}}'$ respectively. 
	We then obtain the desired braiding processes 
	\begin{align}
		\gamma_2 \gamma_1 \gamma_2 |\Psi\rangle = - |\Psi\rangle, \qquad
		\gamma_2' \gamma_1 \gamma_2' |\Psi\rangle = - |\Psi\rangle. 
	\end{align}
	
	\begin{figure}[h!]
		\centering
		\includegraphics[width=0.35\textwidth]{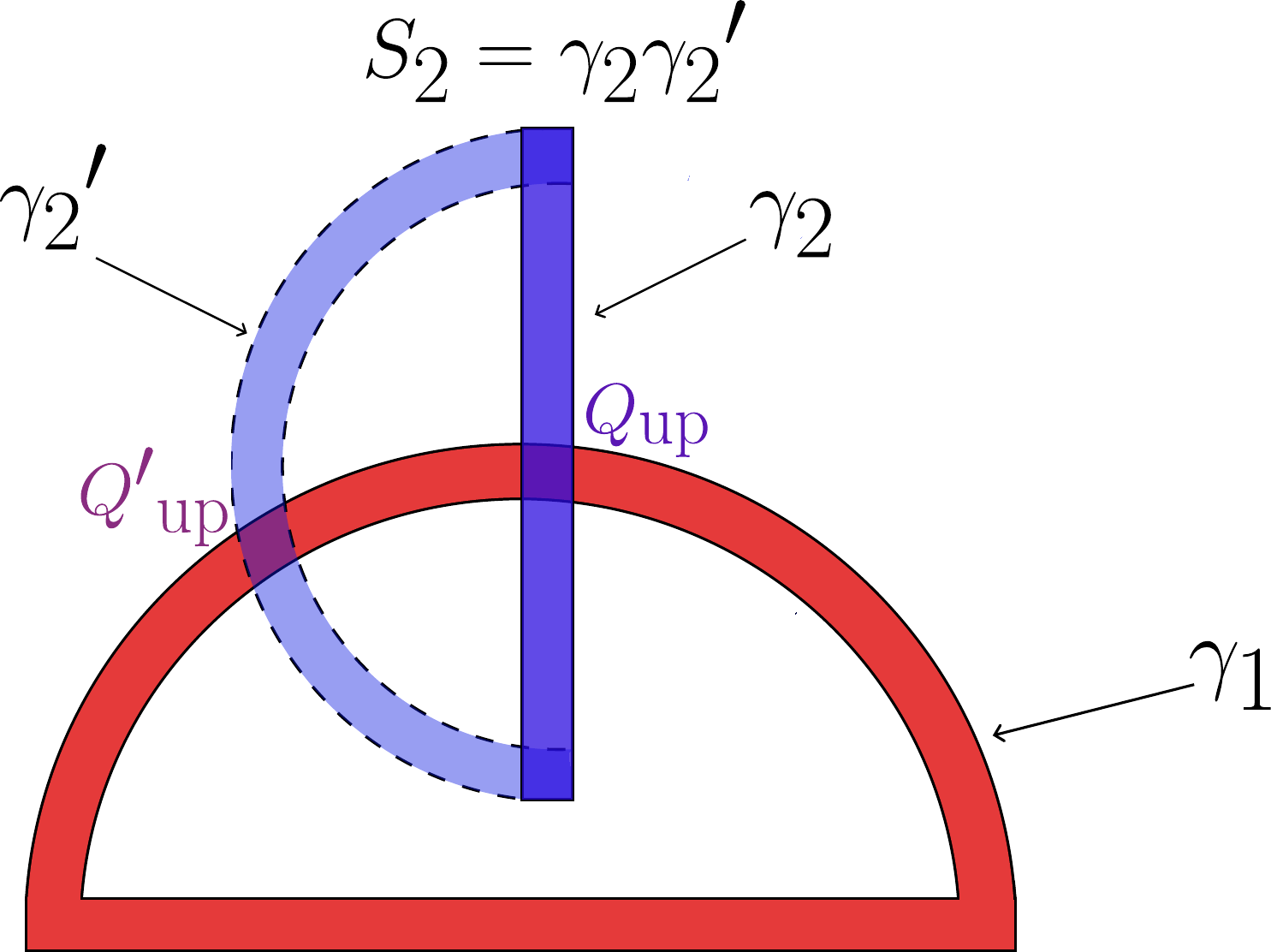}
		\caption{Braiding process using truncated stabilizer operators.}
		\label{fig:braiding}
	\end{figure}
	
	\subsection{Overlap with short-range-entangled states}
	
	The above arguments assume $t=0$. Applying a finite-depth circuit, the correctable squares now have width $\Theta(\frac{d}{w})-\Theta(t)=\Theta(\frac{d}{w})$, and the separation is now $\Theta(w+t)$. 
	With the braiding processes, the same argument as in Lemma \ref{lemma:TCLRE} implies $\ket{\Psi}\not= U_{t}\ket{0}^{\otimes n}$ if $\frac{d}{w} > \Theta(w+t)$. 
	We further lower bound $E_{t}(\Psi)$ using a similar strategy.
	Starting from a mesh $B$ of width $\Theta(w+t)$, consider a family of meshes, each obtained by shifting the previous meshes along the diagonal by $\Theta(w+t)$. 
	For each pair of meshes, we can repeat the above argument and find an anticommuting intersection and a local upper bound as in Lemma \ref{lemma:step2toric}.
	Moreover, each anticommuting intersection must be separated from each other by $\Theta(w+t)$ and is, hence, decoupled from each other.
	The number of such meshes is at least $\Theta(\frac{d}{w(w+t)})$; hence the number of anticommuting intersections is at least $\Theta(\frac{d^2}{w^2(w+t)^2})$.
	This completes the proof of Theorem~\ref{thm_stabilizer}.

	\section{Honeycomb model}\label{sec:fermion}
	
	The original honeycomb model, introduced by Kitaev, is given by~\cite{kitaev2006anyons}:
	\begin{align}
		H_{\text{Kitaev}}= - J_X\sum_{X\text{-links}}
		\figbox{0.2}{fig_2bdyX.pdf}
		- J_Y\!\! \sum_{Y\text{-links}}\!\! \figbox{0.2}{fig_2bdyY.pdf}
		- J_Z\sum_{Z\text{-links}} \figbox{0.2}{fig_2bdyZ.pdf}
	\end{align}
	where qubits reside on vertices of a honeycomb lattice. The model is exactly solvable via mapping to Majorana fermions. For $J_{Z}\gg J_{X},J_{Y}>0$, its low-energy physics can be approximated by the toric code at the leading order in perturbation theory. 
	Hence, from our argument in Sec.~\ref{sec:toric}, we can deduce that the ground state of $H_{\text{Kitaev}}$ in the ``toric code phase'' satisfies $E_{t}(\Psi) > O(n)$ for constant $t$.

	Here we focus on another honeycomb Hamiltonian, denoted as $H_{\text{fermion}}$, that reflects the symmetry properties of the original $H_{\text{Kitaev}}$. 
	The starting point is an observation that the original model $H_{\text{Kitaev}}$ possesses peculiar looplike symmetry operators
	\begin{align}
		[H_{\text{Kitaev}},S_{\hexagon}]=0,\qquad S_{\hexagon} = \figbox{0.25}{fig_honeycombstab.pdf}
	\end{align}
	where $S_{\hexagon}$ can be generated by multiplying all the two-body terms on each hexagon. 
	The ground state $|\Psi_{\text{toric}}\rangle$ of $H_{\text{Kitaev}}$ in the toric code phase ($J_{Z}\gg J_{X},J_{Y}>0$) satisfies $S_{\hexagon}|\Psi_{\text{toric}}\rangle = |\Psi_{\text{toric}}\rangle$ (i.e. it is vortex-free~\cite{kitaev2006anyons}). 
	
	The honeycomb model with emergent fermions, which we address in this paper, can be formulated as a 2D stabilizer Hamiltonian
	\begin{align}
		H_{\text{fermion}} = - \sum_{\hexagon} S_{\hexagon}, \qquad S_{\hexagon} = \figbox{0.25}{fig_honeycombstab.pdf}
	\end{align}
	where qubits reside on vertices of a hexagonal lattice and $\hexagon$ are hexagonal plaquettes. A ground state of $H_{\text{fermion}}$ satisfies  $S_{\hexagon}|\Psi\rangle = |\Psi\rangle$.
	We are interested in quantifying the entanglement of wave functions that live in the ground space of $H_{\text{fermion}}$. 
	
	What are the key differences between $H_{\text{Kitaev}}$ and $H_{\text{fermion}}$?
	First of all, note that the toric code ground state $|\Psi_{\text{toric}}\rangle$ of $H_{\text{Kitaev}}$ is contained in the ground space of $H_{\text{fermion}}$.
	In fact, the ground space of $H_{\text{fermion}}$ is huge. 
	Viewed as a stabilizer code, there are $\frac{n}{2}$ stabilizer generators $S_{\hexagon}$ with one redundancy relation $\prod_{\hexagon} S_{\hexagon}=I$. Hence, the number of logical qubits is $k=1+ \frac{n}{2}$ which suggests an exponentially large ground state space. 
	Second, as a quantum error-correcting code, $H_{\text{fermion}}$ is not fault tolerant. 
	Indeed, logical operators of the code, which must commute with all $S_{\hexagon}$, are generated by two-body Pauli operators:
	\begin{align}\label{eq:XXYYZZ2}
		\figbox{0.3}{fig_2bdyX.pdf}\qquad
		\figbox{0.3}{fig_2bdyY.pdf}\qquad
		\figbox{0.3}{fig_2bdyZ.pdf}
	\end{align}
	and hence, the code distance is $d=2$.

	Because of an exponential ground state degeneracy and a small code distance, one might think that the honeycomb model has a ground state $|\Psi\rangle$ which would have a large overlap with a short-range-entangled state. 
	Contrary to this naive expectation, we will obtain a linear lower bound on $E_{t}(\Psi)$ for any state $|\Psi\rangle$ in the exponentially large ground space of the honeycomb model. 
	
	\begin{theorem}\label{thm:fermion}
		There exists an absolute constant $\alpha>0$ such that for any ground state $|\Psi \rangle$ of the honeycomb model $H_{\text{fermion}}$, if $t\leq O(\sqrt{n})$, then
		\begin{equation}
			E_{t}(\Psi) > \frac{\alpha n}{(t+1)^2} \ .
		\end{equation}
	\end{theorem}
	
	\subsection{Fermion exchange statistics}
	
	A key property of the honeycomb model is that its ground space supports emergent fermions. 
	Consider an open string operator $M$ which is generated by multiplying two-body logical operators (\refeq{eq:XXYYZZ2}). 
	Let us pick an arbitrary ground state $|\Psi\rangle$ and regard it as a ``vacuum'' state with no fermions. 
	Then, $M|\Psi\rangle$ can be interpreted as a state with two emergent fermions at end points of an open string $M$. 
	A stabilizer generator $S_{\hexagon}$ implements a process of creating and annihilating a pair of fermions, hence leaving the vacuum state unchanged $S_{\hexagon}|\Psi\rangle =|\Psi\rangle$.
	These interpretations can be verified by a standard mapping of two-body generators to Majorana fermion operators as in Ref.~\cite{kitaev2006anyons}.
	
	It is worth emphasizing that a state $M|\Psi\rangle$ with a pair of fermions, as well as the vacuum state $|\Psi\rangle$, lives in the ground space of the honeycomb model.
	This is in contrast with the toric code where anyons emerge as excitations that depart from the ground space of the Hamiltonian. 
	
	An immediate challenge is that fermions have trivial braiding statistics (i.e., they are not anyons), and, thus, the argument from Sec.~\ref{sec:braiding} is not applicable.
	In the honeycomb model, this is reflected in that an open string operator $M$ always commutes with a closed string operator $S$ since $S$ is generated by stabilizer generators $S_{\hexagon}$ (see Fig.~\ref{fig:honeycoms}). 
	
	\begin{figure}[h!]
		\centering
		\subfloat{\includegraphics[width = 4cm]{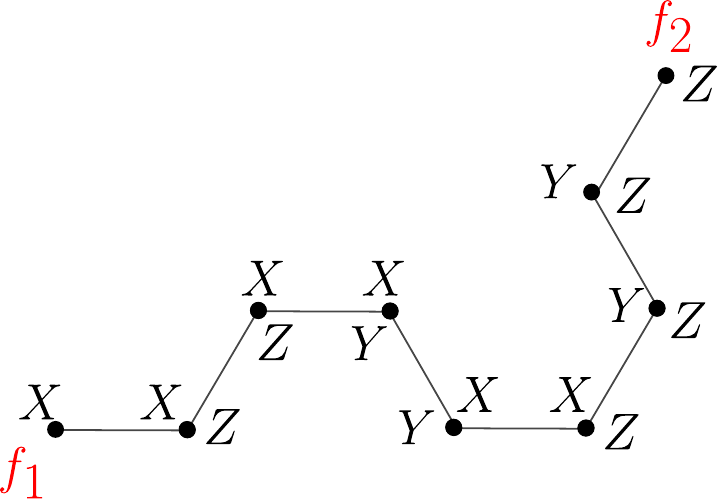}}
		\qquad \qquad 
		\subfloat{\includegraphics[width = 3cm]{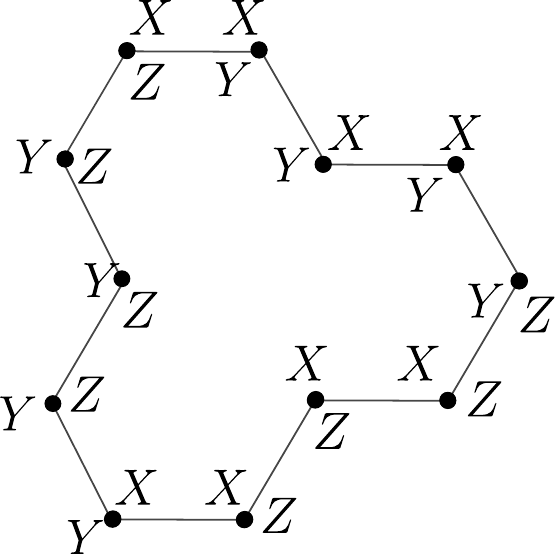}}
		\caption{(Left) An open string operator and (Right) a symmetry operator in the honeycomb model.} 
		\label{fig:honeycoms}
	\end{figure}
	
	Instead, we rely on the fermion exchange statistics to obtain a linear lower bound on $E_{t}(|\Psi\rangle)$. 
	The topological spin of two fermions can be extracted by the hopping operator algebra due to Levin and Wen~\cite{levin2003fermions}, which captures the effect of exchange via creation and annihilation processes:
	\begin{eqnarray}\label{eq:Tshape}
		\bra{\Psi}M_3^\dagger M_2^\dagger M_1^\dagger M_3M_2M_1\ket{\Psi}= e^{i \theta},
	\end{eqnarray}
	where $M_i$ are open string operators as in Fig. \ref{fig:exchange} and $e^{i \theta}=-1$ for fermions. 
	For the honeycomb model, Eq.~(\ref{eq:Tshape}) can be readily verified as three open string operators $M_1, M_2$ and $M_3$ with a common end point anticommute with each other.
	The exchange statistics is stable under deformation by a constant depth circuit $U_t$ by considering dressed string operators and quasilocal fermions. 
	
	\begin{figure}[h!]
		\centering
		\includegraphics[width = 5cm]{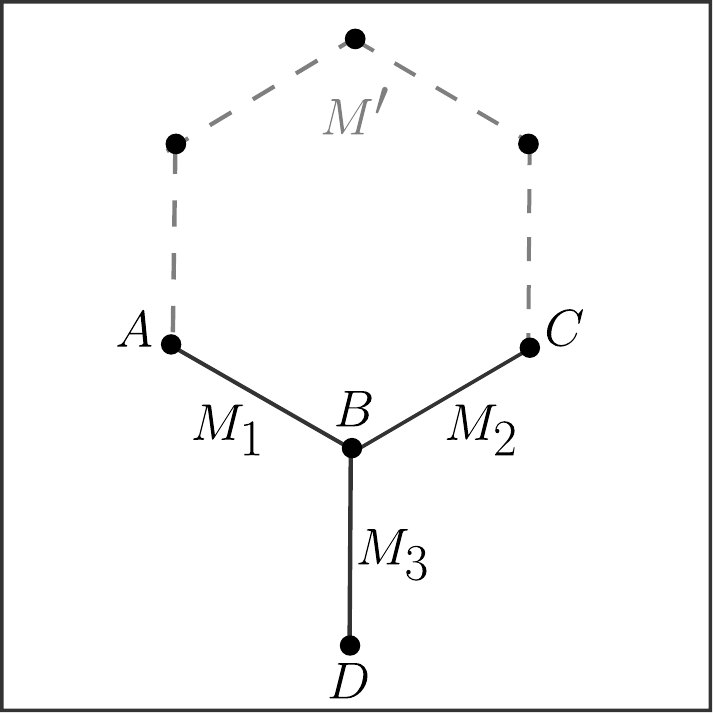}
		\caption{Exchange statistics of two fermions.}
		\label{fig:exchange}
	\end{figure}
	
	\subsection{Long-range entanglement}
	
	We first prove that a ground state of the honeycomb model is long-range-entangled. 
	
	\begin{lemma}\label{lemma:step1_fermion}
		For any ground state $|\Psi\rangle$ of the honeycomb model, we have
		\begin{align}
			U_t |\Psi\rangle \not= \ket{0}^{\otimes n} \label{eq:step1_fermion}
		\end{align}
		for any depth-$t$ unitary $U_t$ with constant $t$.
	\end{lemma}
	
	\begin{proof} 
		Suppose that Eq.~(\ref{eq:step1_fermion}) is false, namely $|\Psi_{t}\rangle \equiv U_t |\Psi\rangle = \ket{0}^{\otimes n}$. 
		Let $M_1,M_2$ and $M_3$ be dressed open string operators satisfying \refeq{eq:Tshape}.
		Using the deformability of string operators, one obtains
		\begin{eqnarray}\label{eq:44}
			M_3 (M_2M_1) |\Psi_t \rangle = |\psi_{BD} \rangle |\psi_{AC} \rangle |0\rangle^{\otimes (ABCD)^c},
		\end{eqnarray}
		where $|\psi_{BD} \rangle$ and $|\psi_{AC} \rangle$ are some wave functions supported on $BD$ and $AC$. 
		
		A crucial observation is that both $|\psi_{BD} \rangle$ and $|\psi_{AC} \rangle$ are factorized:
		\begin{align}
			|\psi_{BD} \rangle|\psi_{AC} \rangle
			= |\psi_{B} \rangle |\psi_D \rangle 
			|\psi_{A} \rangle |\psi_C \rangle.  \label{eq:product}
		\end{align}
		To show this\footnote{
			If string operators can be implemented via short-depth circuits (as in the honeycomb model), this statement can be immediately proven. 
			Namely, since the neighborhoods of regions $B$ and $D$ are initially unentangled, $B$ and $D$ cannot be entangled via a short-depth circuit that implements $M_3$.
		},
		let us construct a deformed string operator $M'$ which connects $A$ and $C$ as in Fig.~\ref{fig:exchange},
		such that the loop operator $M'M_2M_1$ is a product of some stabilizer generators $S_{\hexagon}$,
		and, thus, commutes with all string operators.
		We then have:
		\begin{eqnarray}
			&&M_3 (M_2M_1)\ket{\Psi} \cr
			&=& {M'}^\dagger M' M_3 (M_2M_1)\ket{\Psi} = {M'}^\dagger M_3 (M'M_2M_1)\ket{\Psi} \cr
			&=& {M'}^\dagger (M'M_2M_1) M_3 \ket{\Psi} = (M_2M_1) M_3 \ket{\Psi}.
		\end{eqnarray}
		We can, however, write the right-hand side as $M_2 (M_1M_3) |\Psi_t \rangle$ and decompose it as $|\psi_{AD}\rangle |\psi_{BC} \rangle |0\rangle^{\otimes (ABCD)^c}$ similar to \refeq{eq:44}, implying:
		\begin{align}\label{eq:47}
			|\psi_{BD} \rangle|\psi_{AC} \rangle = |\psi_{AD}\rangle |\psi_{BC} \rangle.
		\end{align}
		The right-hand side implies $B$ and $D$ are unentangled, and so are $A$ and $C$.
		This shows that $|\psi_{BD}\rangle$ and $|\psi_{AC}\rangle$ must be factorized as in Eq.~(\ref{eq:product}). 
		
		By properly choosing the phases of $\ket{\psi_B}$, $\ket{\psi_C}$ and $\ket{\psi_D}$ in order, we can assume
		\begin{equation}
			\begin{aligned}
				M_1 |\Psi_t \rangle &=|\psi_A \rangle |\psi_B \rangle |0\rangle^{\otimes (AB)^c},\\
				M_2M_1 |\Psi_t \rangle &=|\psi_A \rangle |\psi_C \rangle |0\rangle^{\otimes (AC)^c},\\
				M_3 |\Psi_t \rangle &=|\psi_B \rangle |\psi_D \rangle |0\rangle^{\otimes (BD)^c}.
			\end{aligned}
		\end{equation}
		Using the first two equations,
		we obtain $M_2 |\psi_B \rangle |0\rangle^{\otimes (AB)^c} = |\psi_C \rangle |0\rangle^{\otimes (AC)^c}$.
		Therefore, 
		\begin{eqnarray}
			&&M_1 M_2 M_3 |\Psi_t \rangle \cr
			&=& M_1 M_2 |\psi_B \rangle |\psi_D \rangle |0\rangle^{\otimes (BD)^c}       
			= M_1 |\psi_C \rangle |\psi_D \rangle |0\rangle^{\otimes (CD)^c} \cr
			&=&|\psi_A \rangle |\psi_B \rangle |\psi_C \rangle |\psi_D \rangle |0\rangle^{\otimes (ABCD)^c} \cr
			&=& M_3(M_2M_1)|\Psi_t \rangle
		\end{eqnarray}
		which contradicts with \refeq{eq:Tshape} with $e^{i \theta} = -1$. 
	\end{proof}
	
	\subsection{Overlap with short-range-entangled state}
	
	We can prove the following upper bound on the maximum overlap between $\rho_R$ and $|0_R\rangle$. 
	Similar to Lemma \ref{lemma:step2toric}, it shows a ``local GEM" that remains bounded away from zero by an absolute gap that is independent of the subregion size.
	
	\begin{lemma}\label{lemma:step2fermion}
		Let $R$ be any patch of linear size $\geq 8(t+1)$. 
		Given any deformed honeycomb model ground state $|\Psi_t \rangle = U_t |\Psi\rangle$  with a depth-$t$ circuit $U_t$, 
		there exists an absolute constant $\epsilon >0$ such that, in terms of the trace norm,   
		\begin{equation}\label{eq:fermioncbound}
			\big\lVert\rho_R-\ketbra{0_R}\big\rVert_{1}>\epsilon,
		\end{equation}
		where $\rho_{R} = \Tr_{R^c}(|\Psi_t \rangle \langle \Psi_t|)$, and, accordingly, 
		\begin{eqnarray}
			{\cal F}(\ket{0_R}, \rho_{R}) < 1- \frac{1}{4} \epsilon^2=1-\epsilon'
		\end{eqnarray}
		for an absolute constant $\epsilon' > 0$.
	\end{lemma}
	
	The proof can be readily obtained by making the proof of Lemma~\ref{lemma:step1_fermion} more quantitative with each equation replaced with $\approx$ in terms of the trace distance (see the Appendix for the full proof). 
	We emphasize that, similar to Lemma \ref{lemma:step2toric}, all we need to assume on $\rho_R$ is the deformability of string operators (resulting from $1$-form symmetry) and the fermionic exchange statistics (resulting from the anomalous nature of symmetry). 
	In particular, the same result, with the same constant $\epsilon'$, would apply to any state that is invariant under $S_\hexagon$ in the region of interest.

	Finally, we prove a linear lower bound on $E_{t}(\Psi)$ for a honeycomb model ground state by adding up the ``local GEM". 
	Let us divide the whole system into multiple regions $R_{j}$ with $R=\cup_{j=1}^m R_j$. 
	Unlike the toric code, the honeycomb model does not satisfy the decoupling property as in \refeq{eq:decoupling}, since the code distance is small. 
	Instead, we proceed by considering postprojection states sequentially and bound the total projection amplitude.  
	
	Let $\Pi_{j}$ be an operator that projects $\mathcal{H}_{R_{j}}$ to $|0_{R_{j}}\rangle$. 
	Recall that
	\begin{align}
		\mathcal{F}(|\Psi_t\rangle, \ket{0}^{\otimes n}) 
		\leq 
		\Tr( \Pi_R \rho_R ).
	\end{align}
	We can express the right-hand side as (define $\Pi_0 = I$)
	\begin{align}
		\Tr( \Pi_R \rho_R ) = \prod_{j=1}^m \mathcal{F}_j, \quad
		\text{where~}
		\mathcal{F}_j \equiv \frac{\Tr(\Pi_j \cdots \Pi_1 \rho_R  )}{ \Tr(\Pi_{j-1} \cdots \Pi_1 \rho_R ) }.
	\end{align}
	We claim that $\mathcal{F}_j < 1- \epsilon'$ for some $\epsilon' > 0$ and $\forall~1\leq j\leq m$. 
	To prove this, we define a normalized postprojection state by 
	\begin{align}\label{eq:postprojection}
		\rho^{(j)} \equiv \frac{\Pi_{j}\cdots \Pi_1\rho_R\Pi_1 \cdots\Pi_{j}}{\Tr (\Pi_{j}\cdots \Pi_1 \rho_R)},~~(0\leq j\leq m-1),
	\end{align}
	where $\mathcal{H}_{R_1},\cdots, \mathcal{H}_{R_{j}}$ have been projected to $|0\rangle$'s.
	
	Crucially, we observe that since $\cup_{i=1}^j R_j$ does not intersect $R_{j+1}$, the normalized postprojection state $\rho^{(j)}$ still satisfies the deformability of the string operators and the fermion exchange statistics inside the patch $R_{j+1}$, namely
	\begin{align}
		\Tr( M_3^\dagger M_2^\dagger M_1^\dagger M_3 M_2 M_1 \rho^{(j)} ) = -1
	\end{align}
	whenever the $M$s are supported in $R_{j+1}$. 
	Therefore, we can apply Lemma~\ref{lemma:step2fermion} to $\rho^{(j)}$ and region $R_{j+1}$ to obtain
	\begin{equation}
		\mathcal{F}_{j+1} = \mathcal{F}(\rho^{(j)},\Pi_{j+1} )< 1-\epsilon',
	\end{equation}
	where $\epsilon'$ is $j$ indepedent.
	Multiplying everything together, we arrive at
	\begin{align}
		\mathcal{F}(|\Psi_t\rangle, \ket{0}^{\otimes n}) < (1-\epsilon')^{\frac{cn}{(t+1)^2}},
	\end{align}
	which proves Thm~\ref{thm:fermion}. 
	
	\subsection{Mixed state}\label{sec:mixedstate}
	
	Quantum circuit complexity of preparing a mixed state has been studied in quantum information theory in the context of topological quantum memory at finite temperature. Recently, there has been renewed interest in classifications of mixed-state phases in the condensed matter community. Here, we discuss the implication of our results from the perspective of mixed-state phases. 
	
	Following Hastings~\cite{hastings2011topological}, we say that a mixed state $\rho$ is short-range-entangled if $\rho$ can be expressed as a probabilistic ensemble of short-range-entangled states:
	\begin{align}
		\rho = \sum_{j}p_{j}|\psi_j\rangle \langle \psi_j|, \qquad |\psi_j\rangle = U_j |0\rangle^{n}. 
	\end{align}
	For a technical reason, it is conventional to assume that the depth of $U_j$ is at most $t\sim O(\mathrm{polylog}(n))$ in defining short-range entanglement in $\rho$. 
	Previous works showed that two-dimensional commuting projector Hamiltonians and three-dimensional fracton models have short-range-entangled Gibbs states according to this definition~\cite{hastings2011topological,siva2017topological}.
	
	We are interested in the quantum circuit complexity of the maximally mixed state $\rho$ that satisfies $\Tr (\rho S_{\hexagon}) = 1$. 
	For any decomposition $\rho = \sum_{j}p_{j}|\psi_j\rangle \langle \psi_j|$, $\Tr (\rho S_{\hexagon}) = 1$ and $||S_{\hexagon}||=1$ would imply
	\begin{align}
		S_{\hexagon}|\psi_j \rangle= 1 
	\end{align}
	for all $j$. It follows from Theorem \ref{thm:fermion} that each $|\psi_j\rangle$ is long-range-entangled. 
	Hence, we arrive at the following result. 
	
	\begin{corollary} 
		Any mixed state $\rho$ in the honeycomb model $H_{\text{fermion}}$ ground state subspace is long-range-entangled. 
	\end{corollary}

	Finally, it will be useful to generalize the notion of depth-$t$ GEM to mixed states. 
	A natural generalization is to consider the maximum fidelity between the state of interest and an ensemble of short-range-entangled states
	\begin{align}
		E_{t}(\rho) = -  \max_{\sigma \in \text{SRE}(t)} \log_2 \mathcal{F}(\rho,\sigma).
	\end{align}
	Here $\text{SRE}(t)$ denotes the set of mixed states that are representable as ensembles of depth-$t$ short-range-entangled pure states.

	\begin{corollary}
		For a symmetric mixed state $\rho$ in the honeycomb model, if $t\leq O(\sqrt{n})$, then 
		\begin{equation}
			E_{t}(\rho) > \frac{\alpha n}{(t+1)^2},
		\end{equation}
		where $\alpha >0$ is an absolute constant. 
	\end{corollary}
	
	\begin{proof}
		Given mixed state $\rho$, consider a quantum channel $\mathcal{N}$ that simultaneously measures stabilizer generators $S$ and records the outcomes.
		Denote the outcome probability distributions as $P_{\rho}(s_{1},\cdots, s_{n/2})$ where $s_{j}=\pm1$ are syndrome values of $S_{\hexagon}$. 
		Define $P_\sigma$ similarly.
		Because of the monotonicity of the fidelity under quantum channels, we have 
		\begin{align}
			\mathcal{F}(\rho,\sigma) \leq \mathcal{F}(\mathcal{N}(\rho),\mathcal{N}(\sigma)) = \mathcal{F}(P_{\rho},P_{\sigma}).
		\end{align} 
		Assuming that $\rho$ is a symmetric state: $\Tr( S_{\hexagon}\rho)=1$, we have $P_{\rho}(+1,\cdots,+1)=1$ and $P_{\rho}(s_1,\cdots,s_{n/2})=0$ for any other $(s_1,\cdots,s_{n/2})$.
		Assuming $\sigma \in \text{SRE}(t)$, we can write it as $\sigma = \sum_{j}p_{j}|\psi_j\rangle \langle \psi_j|$ where $\ket{\psi_j}$, are short-range entangled.
		Then, we have 
		\begin{equation}\label{eq:64}
			\begin{aligned}
				\mathcal{F}(P_{\rho},P_{\sigma}) =& P_{\sigma}(+1,\cdots, +1) = \Tr(\Pi_S \sigma) \\
				=&\sum_{j}p_j \Tr( \Pi_S |\psi_j\rangle\langle \psi_j| ),
			\end{aligned} 
		\end{equation}
		where $\Pi_S$ is a projection to the symmetric subspace\footnote{One can get the same inequality by considering the projection measurement $\{\Pi_S,1-\Pi_S\}$.}. 
		
		Finally, we claim 
		\begin{align}
			\Tr( \Pi_S |\psi_j\rangle\langle \psi_j| )< e^{-\Theta(n)}. \label{eq:mixed-fermion}
		\end{align}
		To prove this, consider a normalized symmetric pure state 
		\begin{align}
			|\Psi_j\rangle = \frac{\Pi_S|\psi_j\rangle}{\sqrt{ \langle \psi_j | \Pi_S| \psi_j \rangle }}. 
		\end{align}
		It follows from Theorem~\ref{thm:fermion} that $|\langle \Psi_j | \psi_j\rangle|^2 < e^{-\Theta(n)}$, which leads to Eq.~\eqref{eq:mixed-fermion}. 
		The desired inequality
		\begin{align}
			\mathcal{F}(\rho,\sigma) < e^{-\Theta(n)}
		\end{align}
		then follows from \refeq{eq:64}.
	\end{proof}
	
	\section{Discussions}\label{sec:discussion}
	
	We comment on a few open problems. 
	
	\begin{enumerate}[1)]
		
		\item A similar quantitative characterization of entanglement may be obtained for symmetry-protected topological order by considering an overlap with trivial wave functions that are prepared by short-depth symmetric circuits (i.e. symmetry-protected GEM). 
		
		\item The GEM can quantitatively distinguish genuine long-range entanglement from ``trivial'' long-range entanglement. For instance, we have $E_t(\text{toric})\sim \Theta(n)$ and $E_t(\text{GHZ})\sim O(1)$.
		One potential manifestation of different GEM scaling is (in)stability of long-range entanglement under local decoherence. 
		It will be interesting to make this intuition concrete by establishing a connection between fault tolerance and GEM. 
		
		\item Studies of depth-$t$ GEM for critical systems (those described by conformal field theories) may also be an interesting research avenue. 
		
		\item Ground states of a topologically ordered Hamiltonian can be generically prepared by $O(L)$-depth quantum circuits, and, as such, have $E_{t}(\Psi)\approx 0$ for $t\sim \Omega(L)$. 
		In a system of $n$ qubits, however, there exist quantum states whose preparation requires $e^{O(n)}$-depth quantum circuits. 
		An example of a quantum state with exponential complexity is a Haar random state. 
		It will be interesting to study the behavior of $E_{t}(\Psi)$ for states prepared from local random unitary circuits. 
		
		\item Recently, there has been significant progress in understanding noninvertible symmetries. It is interesting to study which classes of noninvertible symmetries require long-range entanglement. 
		In fact, we expect that some aspects of noninvertible symmetries and anomalous symmetries are fundamentally akin to each other and can be treated on an equal footing.
		For example, one can implement some anomalous symmetries locally by quantum channels, in a spirit similar to Ref.~\cite{okada2024non}.

	\end{enumerate}

	\begin{acknowledgments} 
		We thank Timothy Hsieh, Han Ma, Zi-Wen Liu, and Chong Wang for useful discussions and comments. 
		We especially thank Sergey Bravyi for collaboration on a related work \cite{bravyi2024much} and many illuminating discussions.
		Research at Perimeter Institute is supported in part by the Government of Canada through the Department of Innovation, Science and Economic Development and by the Province of Ontario through the Ministry of Colleges and Universities. 
		This work is supported by the Applied Quantum Computing Challenge Program at the National Research Council of Canada.

		\textit{Note added}: While we are finalizing the manuscript, version 2 of Ref.~\cite{wang2023intrinsic} appeared on arXiv in May 2024.
		This updated version establishes that a mixed state with anomalous $1$-form symmetry exhibits long-range entanglement.
		Version 1 of Ref.~\cite{wang2023intrinsic}, which appeared on arXiv in July 2023, does not contain this result. 
	\end{acknowledgments} 
	
	\bibliography{anyonbib.bib}
	
	\appendix
	\section{Proof of Lemma \ref{lemma:step2fermion}}\label{app:prove6}
	Here, we prove Lemma \ref{lemma:step2fermion}, copied below for convenience.
	\begin{lemma_copy}{lemma:step2fermion}
		Let $R$ be any patch of linear size $\geq 8(t+1)$. 
		Given any deformed honeycomb model ground state $|\Psi_t \rangle = U_t |\Psi\rangle$  with a depth-$t$ circuit $U_t$, 
		there exists an absolute constant $\epsilon >0$ such that, in terms of the trace norm,   
		\begin{equation}
			\big\lVert\rho_R-\ketbra{0_R}\big\rVert_{1}>\epsilon,
		\end{equation}
		where $\rho_{R} = \Tr_{R^c}(|\Psi_t \rangle \langle \Psi_t|)$, and, accordingly, 
		\begin{eqnarray}
			{\cal F}(\ket{0_R}, \rho_{R}) < 1- \frac{1}{4} \epsilon^2=1-\epsilon'
		\end{eqnarray}
		for an absolute constant $\epsilon' > 0$.
	\end{lemma_copy}
	
	\begin{proof}
		Similar to the proof of Lemma \ref{lemma:step2toric}, we use $\approx$ to denote closeness in the trace distance
		\begin{equation}
			\rho \approx \sigma \iff \norm{\rho-\sigma}_1< O(\epsilon).
		\end{equation}
		For pure states, we use two different notions of closeness:
		\begin{equation}
			\begin{aligned}
				\psi_1\approx \psi_2 &\iff \norm{\psi_1\rangle \langle \psi_1 | - |\psi_2 \rangle \langle \psi_2|}_1 < O(\epsilon),\\
				|\psi_1 \rangle \sim |\psi_2 \rangle &\iff \norm{|\psi_1 \rangle - |\psi_2 \rangle  } < O(\epsilon).
			\end{aligned}
		\end{equation}
		Here, the norm in the second line is the vector norm in Hilbert space. 
		While $\approx$ ignores the phase information, the approximation $\sim$ cares about the phases.
		It is clear that $\psi_1\approx \psi_2$ if and only if we can choose phases so that $\psi_1\sim \psi_2$.

		Now, assuming $\rho_R\approx \ket{0_R}\bra{0_R}$ in a region $R$, let us derive a contradiction. In the following, $M_i$ are dressed operators, and we omit the subscript $t$.

		First, the same as Eq.~(\ref{eq:puredecomposition}) and its proof, a string operator creates two excitations at end points which are approximately decoupled from the rest of the system:
		\begin{equation}\label{eq:66}
			M_3\ket{0_R} \approx |\psi\rangle_{BD} |0\rangle^{\otimes (BD)^c}.
		\end{equation}
		With a similar equation for $M_2M_1$, we get the approximate version of \refeq{eq:44}:
		\begin{equation}\label{eq:approx1}
			\begin{aligned}
				M_3 (M_2 M_1) \ket{0_R}
				\approx &M_3 |\psi \rangle_{AC} |0\rangle^{\otimes (AC)^c} &\\
				\approx &|\psi \rangle_{AC} |\psi\rangle_{BD} |0\rangle^{\otimes (ABCD)^c}.
			\end{aligned}
		\end{equation}
		Here, in the second approximation, we use $M_3 (|0\rangle \langle 0|)^{\otimes (AC)^c} M_3^\dagger = \Tr_{AC} (M_3 |0_R\rangle \langle 0_R | M_3^\dagger ) $ and \refeq{eq:66}.
		Also, similarly, we have
		\begin{equation}\label{eq:approx2}
			M_2(M_1M_3)\ket{0_R} \approx |\psi \rangle_{BC} |\psi \rangle_{AD} |0\rangle^{\otimes (ABCD)^c}.
		\end{equation}

		Next, we show that both $|\phi \rangle_{AC}$ and $|\phi\rangle_{BD}$ approximately factorize.
		From $\rho_R\approx\ket{0_R}\bra{0_R}$ and $M_3M_2M_1 \rho_R= M_2M_1M_3 \rho_R$, we get 
		$M_3M_2M_1 \ket{0_R}\approx M_2M_1M_3 \ket{0_R}$.
		Now taking $\text{Tr}_{BC}$ on Eqs.~(\ref{eq:approx1}) and (\ref{eq:approx2}) and using the monotonicity of the trace distance under partial trace, we get: 
		\begin{eqnarray}
			\phi_{A} \otimes \phi_D \approx \ket{\psi}\bra{\psi}_{AD},
		\end{eqnarray}
		where $\phi_A = \text{Tr}_C (|\psi \rangle\langle \psi|_{AC})$ and $\phi_D = \text{Tr}_B (|\psi \rangle\langle \psi|_{BD})$. 
		It follows that $\phi_A$ and $\phi_D$ are approximately pure states as measured by the purity\footnote{
			For any two states $\sigma_1$ and $\sigma_2$, we have 
			$\Tr(\sigma_1^2-\sigma_2^2)\leq 2\norm{\sigma_1-\sigma_2}_\infty \leq 2\norm{\sigma_1-\sigma_2}_1$.
		}:
		\begin{equation}
			\Tr\phi_A^2\cdot\Tr\phi_D^2 =\Tr\left((\phi_A \otimes \phi_D)^2 \right)=1-O(\epsilon),
		\end{equation}
		which further implies that they are close to pure states in the trace norm\footnote{
			For any $\sigma$, denote the eigenvalues as $\lambda_i$, then $\lambda_{\max}=\sum_i\lambda_{\max}\lambda_i\geq\sum_i\lambda_i^2=\Tr(\sigma^2)$. 
			Hence $\norm{\sigma-P_{\max}}_1=2(1-\lambda_{\max})\leq 2(1-\Tr(\sigma^2))$ where $P_{\max}$ is the projector to the eigenspace corresponding to the largest eigenvalue.
		}.
		Therefore, we arrive at the approximate version of \refeq{eq:product}:
		\begin{align}
			M_2 M_1 M_3 \ket{0_R}
			\approx& M_3 M_2 M_1 \ket{0_R}
			\\
			\approx& |\psi \rangle_A |\psi \rangle_C |\psi \rangle_B |\psi \rangle_D |0\rangle^{\otimes (ABCD)^c}. \notag
		\end{align}

		Finally, we can derive a contradiction by considering the phase.
		The proof is the same as that of Lemma \ref{lemma:step1_fermion}. 
		We simply replace $=$ by $\sim$, which keeps the phase information.    
	\end{proof}
	\clearpage

\end{document}